\newif\ifSODA
\SODAfalse

\ifSODA
\documentclass[twoside,leqno,twocolumn]{article}  
\usepackage{ltexpprt} %,balance
\else
\documentclass[11pt]{article}  
\usepackage{fullpage,amsthm}
\fi

\usepackage
{
        amssymb,
        amsmath,
        tikz,
        pgfplots,
}

\usepackage{graphicx}

\ifSODA
\makeatletter\def\@endproof{\hfill\vbox{\hrule height.2pt\hbox{\vrule width.2pt height5pt \kern5pt \vrule width.2pt} \hrule height.2pt}\outerparskip 0pt\endtrivlist}
\def\paragraph{\@startsection  
 {paragraph}{4}{\parindent}{5pt}{-5pt}{\normalsize\bf}}
\makeatother
\newcommand {\texorpdfstring} [2] {{#1}}
\else
\usepackage[unicode,colorlinks=true,citecolor=black,filecolor=black,linkcolor=black,urlcolor=black,pdfpagelabels,plainpages=false]{hyperref}
\usepackage{bookmark}
\fi

\def\clap#1{\hbox to 0pt{\hss#1\hss}}

\def\mathrlap{\mathpalette\mathrlapinternal}

\def\mathrlapinternal#1#2{%
\rlap{$\mathsurround=0pt#1{#2}$}}

\newcommand {\set}   [1] {\left\{ #1 \right\}}
\newcommand {\brc}   [1] {\left(#1\right)}

\newcommand {\Exp}       {\mathbb{E}}
\newcommand {\Prob}  [1] {\Pr \brc{#1 }}

\newcommand {\E}     [1] {\Exp\left[#1\right]}

\ifSODA\else
\newtheorem{theorem}{Theorem}[section]
\newtheorem{lemma}[theorem]{Lemma}

\newtheorem{corollary}[theorem]{Corollary}
\newtheorem{Definition}[theorem]{Definition}

\newtheorem{fact}[theorem]{Fact}

\newtheorem*{question*}{Question}
\newtheorem*{definition*}{Definition}

\fi

\newtheorem{claim}{Claim}[section]
\newtheorem{remark}{Remark}[section]

\newcommand{\maxcut}{\textsc{Max Cut}~}
\newcommand{\N}{N}

\title{Bilu--Linial Stable Instances of Max Cut and Minimum Multiway Cut}
\author{Konstantin Makarychev\\Microsoft Research
\and Yury Makarychev\thanks{Supported by NSF CAREER award CCF-1150062 and NSF grant IIS-1302662. 
Work done in part while visiting Microsoft Research.}\\TTIC
\and Aravindan Vijayaraghavan\thanks{Supported by the Simons Postdoctoral Fellowship. Work done in part while visiting Microsoft Research.}\\CMU}
\date{}
\begin{document}
\maketitle
\begin{abstract}
We investigate the notion of stability proposed by Bilu and Linial. We obtain an \textit{exact} polynomial-time  algorithm for
$\gamma$-stable Max Cut instances with $\gamma \geq c\sqrt{\log n}\log\log n$ for some absolute constant $c > 0$. 
Our algorithm is robust: it never returns an incorrect answer;
if the instance is $\gamma$-stable, it finds the maximum cut, otherwise, it either finds the maximum cut or certifies that the instance is not $\gamma$-stable. We prove that there is no \textit{robust} polynomial-time algorithm for $\gamma$-stable instances of Max Cut
when $\gamma < \alpha_{SC}(n/2)$, where $\alpha_{SC}$ is the best approximation factor for Sparsest Cut with non-uniform demands. That suggests that solving $\gamma$-stable instances with $\gamma =o(\sqrt{\log n})$ might be difficult or even impossible.

Our algorithm is based on semidefinite programming. We show that the standard SDP relaxation for Max Cut 
(with $\ell_2^2$ triangle inequalities) is integral if 
$\gamma \geq D_{\ell_2^2\to \ell_1}(n)$, where $D_{\ell_2^2\to \ell_1}(n)$ is the least distortion with which every $n$ point metric space of negative type embeds into $\ell_1$. On the negative side, we show that the SDP relaxation is not integral when $\gamma < D_{\ell_2^2\to \ell_1}(n/2)$.
Moreover, there is no tractable convex relaxation for $\gamma$-stable instances of Max Cut
when $\gamma < \alpha_{SC}(n/2)$. 

Our results significantly improve previously known results. The best previously known algorithm for $\gamma$-stable instances of Max Cut 
required that $\gamma \geq  c\sqrt{n}$ (for some $c > 0$)  [Bilu, Daniely, Linial, and Saks]. No hardness results were known for the problem.

Additionally, we present an exact robust polynomial-time algorithm for $4$-stable instances of Minimum Multiway Cut.

We also study a relaxed notion of \textit{weak stability} and present algorithms for weakly stable instances of Max Cut and Minimum Multiway Cut.
\end{abstract}

\ifSODA\else
\setcounter{page}{0}
\thispagestyle{empty} %%suppress the first page number (which is "0")
\pagebreak
\fi
%\pagebreak
%\thispagestyle{empty}
%\tableofcontents

\section{Introduction} \label{sec:intro}
Empirical evidence suggests that many discrete optimization problems like clustering and partitioning
are much easier in practice than in the worst case. Even though these problems are usually provably 
hard in the worst case, we can still try to design algorithms that 
work well on instances that we encounter in practice. To do so, we need a good mathematical model 
for such instances. 

There are several approaches to modeling real-life instances. Perhaps, a more classical approach dating back
to early 1980's is to assume that real-life instances come from a random or ``semi-random'' distribution~\cite{BS,FKra,FK,KMM,MMV1,MMV2}. To learn more about this approach,
we refer the interested reader to our previous work~\cite{MMV1} on semi-random instances of graph partitioning and references therein.
An alternative approach, which we study here, is to identify certain structural properties that
``interesting'' instances (or ``practically interesting instances''~\cite{BDLS}) must satisfy, and then assume that instances arising in practice 
satisfy them. One such
property was proposed by Bilu and Linial~\cite{BL}.

Bilu and Linial~\cite{BL} introduced a notion of {\em stability of instances} for discrete optimization problems. They argue that interesting instances have stable solutions: the optimal solution does not change upon small perturbations. For example, a clustering instance that is meaningful should have a solution that stands out. This solution should remain optimal even if the edge weights are slightly inaccurate or noisy.
As Balcan, Blum and Gupta~\cite{BBG} argue, the real goal of solving a clustering problem is often to obtain the correct ``target''
clustering, and so the objective function serves only as a proxy. In this case, if the edge weights, which may be rough estimates of how similar or dissimilar the endpoints are, are imprecise, then the solution is meaningful only when the instance is stable. 
Here is a formal definition of $\gamma$-stability~\cite{BL}.

\begin{Definition}[$\gamma$-Stability~\cite{BL}]
Consider an instance of a graph optimization problem on $n$ vertices, defined by the matrix of non-negative edge weights $w$. We say that the instance is $\gamma$-stable if there is an optimal solution which remains optimal, even when any subset of the edge weights are increased by a factor of at most $\gamma$. 
\end{Definition} 

We note that prior to work of Bilu and Linial~\cite{BL}, Balcan, Blum and Gupta~\cite{BBG} introduced and studied a somewhat 
similar but different notion of \textit{approximation--stability} for clustering problems like $k$-means and $k$-median;
later Awasthi, Blum and Sheffet~\cite{Awasthietal}, Balcan and Liang~\cite{BalcanL}, and Reyzin~\cite{Rey} studied a related notion of \textit{perturbation resilience} for these center--based clustering problems.

%As in Bilu and Linial~\cite{BL} and Bilu, Daniely, Linial, and Saks~\cite{BDLS}, 
We study stable instances of Max Cut and Minimum Multiway Cut, and propose a general technique for solving 
stable instances of graph partitioning problems (see Section~\ref{sec:Discussion}). However, to be more specific,
we focus  our exposition on Max Cut --- the problem that was previously studied by Bilu and Linial~\cite{BL} and Bilu, Daniely, Linial, and Saks~\cite{BDLS}.
In Max Cut, we are given a weighted graph $G(V,E,w)$ on $n$ vertices with an adjacency matrix $w$. Our goal is to find a cut $(S, V \setminus S)$ in the graph with the maximum weight of edges crossing it
$$\maxcut(G)= \max_{S \subseteq V(G)} \sum_{e \in E(S, V\setminus S)} w_e.$$
Max Cut is one of the classic NP-hard problems~\cite{GJ}. It is NP-hard to approximate within a factor of $17/16$ \cite{Hastad,TSSW}.
Goemans and Williamson~\cite{GW} gave a semidefinite programming based algorithm that achieves a $0.878$ approximation ratio for Max Cut. Khot, Kindler, Mossel and O'Donnell~\cite{KKMO} showed that this is the best possible approximation ratio in the worst-case, assuming the Unique Games Conjecture \cite{Khot}.

In the work that introduced $\gamma$-stability, Bilu and Linial~\cite{BL} designed an algorithm for $\gamma$-stable instances of Max Cut with $\gamma \geq cn$ (for some absolute constant $c$). The aim of 
the algorithm is to find the exact optimal solution. This solution is unique (because of stability) and corresponds to the ``true'' partitioning we want to find. Finding just a good approximation for $\gamma$-stable instances of Max Cut is easy, 
since $\gamma$-stable instances of Max Cut are almost bipartite, and for almost bipartite graphs, the algorithm of Goemans and Williamson~\cite{GW} 
returns a solution in which almost all edges are cut (see~\cite{BL} for more details). Bilu, Daniely, Linial, and Saks~\cite{BDLS} gave an algorithm for $\gamma$-stable Max Cut instances with $\gamma  \geq c \sqrt{n}$ (for some absolute constant $c$). Both papers also gave better algorithms for
stable instances that satisfy some extra conditions (see~\cite{BL} and~\cite{BDLS}). 
 
\paragraph{Our Results.} In this work, we give an algorithm that solves $\gamma$-stable instances of Max Cut with $\gamma \geq c\sqrt{\log n} \log\log n$ (for some absolute constant $c$). We also study the classic partitioning problem of finding the Minimum Multiway Cut (see Section~\ref{sec:muliway} for the definition) and give an algorithm for $\gamma$-stable instances of it, with $\gamma\geq 4$. 
Our result for Max Cut is an
exponential improvement over previous results. Our algorithms are {\em robust} (in the notion of 
Raghavan and Spinrad~\cite{RS}):
\begin{enumerate}
\item If the instance is $\gamma$-stable, the algorithm finds the unique optimal solution.
\item If the instance is not $\gamma$-stable, the algorithm either finds an optimal solution, or a (polynomial-time verifiable) certificate that proves that the instance is not $\gamma$-stable. 
\end{enumerate}   
In other words, our algorithm is always correct: when we claim to output the maximum cut (minimum multiway cut), 
we can guarantee its optimality; else we identify that the given graph is not sufficiently stable. This is a very desirable property for algorithms we want to use in practice,
since we may only assume that real-life or important instances are stable (or satisfy other properties), but we cannot be 
completely certain that they indeed are. When we use robust algorithms, we cannot get a suboptimal solution even if our assumptions are not quite correct.
Note that previous algorithms for stable instances of Max Cut~\cite{BL,BDLS}, and Clustering~\cite{BBG,BalcanL,Awasthietal} are not robust. That is, if the instance is not $\gamma$-stable, the previous algorithms can output a suboptimal solution without notifying us that the 
solution is suboptimal.

Our algorithms use that our SDP and LP relaxations for Max Cut and Minimum Multiway Cut, respectively, are integral
when $\gamma$ is sufficiently large. 
For Max Cut, we prove that the standard SDP relaxation with triangle inequalities is integral for $(c\sqrt{\log n}\log\log n)$-stable instances.
We remark that we are unaware of other natural settings when the semidefinite program becomes integral and the corresponding linear program does not! 

We also present algorithms that work for the same values of $\gamma$ with a more relaxed notion of stability 
which we call {\em weak stability}.
% ($\gamma \geq c\sqrt{\log n}\log\log n$ for Max Cut and $\gamma \geq 4$ for Minimum Multiway Cut). 
The optimal solution of every perturbed instance of a weakly stable instance, is close to the optimal solution
of the original instance, but may not be exactly the same (see Section~\ref{sec:weakstability} for details). We believe that $\gamma$-weak stability may be a more realistic assumption than $\gamma$-stability in practice. Bilu and Linial~\cite{BL} mentioned weakly stable instances in the introduction to their paper (without formally defining them), and proposed to study them in the future.
% Moreover, instances that satisfy the notion of approximation--stability of Balcan, Blum and Gupta~\cite{BBG} are also weakly stable. 
Our algorithms for $\gamma$-weakly stable instances are not robust.
%Additionally, we study stable instances of Minimum Multiway Cut. We give a robust polynomial-time algorithm for $4$-stable instances of the problem.
%We also present an algorithm for $4$-weakly stable instances of Minimum Multiway Cut.

Our result for weakly stable instances of Max Cut uses an approximation algorithm for Sparsest Cut with non-uniform demands
as a black box. In particular, the result implies that if there is an $\alpha(n)$-approximation algorithm for Sparsest Cut (which finds approximately not only the value but also a solution to Sparsest Cut) 
then there is an exact algorithm for $(1+\varepsilon) \alpha(n)$-stable instances of Max Cut. (For simplicity of exposition,
$\varepsilon$ is fixed in our proof; in general, $\varepsilon$ can be sub-constant, the running time
is proportional to $1/\varepsilon$.)

Finally, we present a general approach to solving stable instances of graph partitioning problems.

\paragraph{Negative Results.} We supplement our algorithmic results, by showing that any robust algorithm for Max Cut that works for better values of $\gamma$ would result in a similar improvement in the worst-case approximation for non-uniform Sparsest Cut. Moreover, we also show that the SDP is not integral if $\gamma < D_{\ell_2^2\to \ell_1}(n/2)$ (where $D_{\ell_2^2\to \ell_1}(n)$ is the least distortion with which every $n$-point $\ell_2^2$ space can be embedded into $\ell_1$). While our algorithmic results give algorithms for sufficiently stable instances, our reduction from non-uniform Sparsest Cut suggests that it may be hard to obtain robust algorithms for Max Cut that work for $O(1)$-stable instances. 
Finally, we describe a very strong negative result for the Max $k$-Cut problem when $k \geq  3$.
We show that for every function $\gamma(n)$ there is no exact polynomial-time algorithm for $\gamma(n)$-stable instances of the problem
unless $NP=RP$.

We note that our positive results for Max Cut also apply to the problem of clustering points into two clusters (or, equivalently, to Max Cut
with positive and negative weights). Our negative result for Max $k$-Cut, on the other hand, shows that there is no exact algorithm for
the problem of clustering points into $k$ clusters when $k\geq 3$ unless $RP=NP$ (see Appendix~\ref{sec:corrclust} for details). 

\subsection{Overview of Techniques}

The main technical component of our work is showing that the standard SDP relaxation with triangle inequalities
is {\em integral} for $\gamma$-stable instances of Max Cut, when $\gamma\geq D_{\ell_2^2\to \ell_1}(n)$, where
$D_{\ell_2^2\to \ell_1}(n) = O(\sqrt{\log n} \log\log n)$ is the least distortion with which every $n$-point $\ell_2^2$ space can be embedded into $\ell_1$.
Hence, when the instance is $\gamma$-stable (for sufficiently large $\gamma$), we can read off the optimal solution from the SDP; otherwise, the non-integrality of the SDP certifies that the instance is not $\gamma$-stable. This gives a {\em robust} algorithm for instances with $\gamma \geq c
\sqrt{\log n}\log \log n$.

Loosely speaking, we prove that if the SDP solution is not integral, then the integral optimal solution can be improved,
possibly after changing the weights of some \ifSODA\linebreak \newpage\noindent\fi edges.
Given the SDP solution $\{\bar{u}\}$ (a vector $\bar{u}$ for each vertex $u$) and the optimal integral solution $(S, V\setminus S)$,
we define a new configuration of vectors $\{\hat{u}\}$: we keep $\hat{u} = \bar{u}$ if $u\in S$; and replace $\bar{u}$
with $\hat{u} = -\bar{u}$ if $u\in V\setminus S$. If $\{\bar{u}\}$ is the integral solution corresponding to $(S, V\setminus S)$,
then all vectors $\hat{u}$ are equal. Otherwise, if the SDP is not integral, \textit{not all vectors $\hat{u}$ are equal}. Then we embed 
these vectors into $\ell_1$, and obtain a distribution
of cuts $(A', V\setminus A')$. It turns, out that if the distortion of the embedding was 1, then we would be able to 
improve the quality of the optimal solution by picking a random cut $(A', S\setminus A')$ and moving vertices in $S\cap A'$ to $V\setminus S$ and vertices in $(V\setminus S)\cap A'$ to $S$
(see Figure~\ref{fig:cuts}). Of course, the distortion may be much larger than 1. Then, we show how to compensate
for this distortion using the $\gamma$-stability.
We can first change the weights of some edges by a factor at most $D_{\ell_2^2\to \ell_1}(n)$, and only then 
move the vertices to improve the quality of the solution. Hence, we get a contradiction if $\gamma \geq D_{\ell_2^2\to \ell_1}(n)$.

Our algorithm for weakly stable instances starts with an approximate solution and then iteratively improves the quality of the solution 
using the algorithm for non-uniform Sparsest Cut by Arora, Lee, and Naor~\cite{ALN} as a subroutine.

\subsection{Outline}
In Section~\ref{sec:prelims}, we introduce the formal definitions of stability and some preliminaries including the semidefinite program (SDP) we use in our algorithm. Then, in Section~\ref{sec:robustalgo}, we describe our robust algorithm for $\gamma$-stable instances of Max Cut. 
In Section~\ref{sec:negative}, we present evidence suggesting that obtaining algorithms for better values of $\gamma$ may not be easy. We first give a reduction from non-uniform Sparsest Cut (in Section~\ref{sec:hardness}), which shows that any robust algorithm with better guarantees would lead to a similar improvement for non-uniform Sparsest Cut. Then we show in Section~\ref{sec:intgap} that the SDP is not integral for smaller values of $\gamma$. 
In Section~\ref{sec:muliway}, we present our algorithm for $4$-stable instances of Minimum Multiway Cut.
In Section~\ref{sec:weakstability}, we introduce a more general notion of weak stability and obtain similar guarantees in
this setting.

We describe our results for Max $k$-Cut and Correlation Clustering
in Sections~\ref{sec:max-k-cut} and Appendix~\ref{sec:corrclust}, respectively.
Finally, we outline a general approach for solving stable instances of graph partitioning problems in Section~\ref{sec:Discussion}. 

% robsutness was defined in http://dl.acm.org/citation.cfm?id=365507

\section{Preliminaries}\label{sec:prelims}
%%\paragraph{Instance Stability.} 
We start with formally defining the notion of Bilu--Linial stability for Max Cut instances. Following~\cite{BL}, we give two equivalent
definitions (see Proposition 2.1 in~\cite{BL}).

\ifSODA
\pagebreak
\fi

\begin{Definition}[Bilu and Linial~\cite{BL}]\label{def:stability1}
\ifSODA \ \\[0.2em] \fi
Let $G = (V, E, w)$ be a weighted graph with edge weights $w(e)$ and let $\gamma > 1$. A weighted graph $G' = (V, E, w')$ is a $\gamma$-perturbation of $G$ 
if for every $(u,v) \in V$,
$$w(u,v) \leq w'(u,v) \leq \gamma \cdot w(u,v).$$

We say that $G$ is a $\gamma$-stable instance of Max Cut if there is a unique cut which 
forms a maximal cut for every $\gamma$-perturbation $G'$ of $G$.
\end{Definition}
\begin{Definition}[Bilu and Linial~\cite{BL}]\label{def:stability2}
Let $\gamma \geq 1$. A weighted graph $G$ graph with maximal cut $(S, \bar S)$ is $\gamma$-stable instance of Max Cut
if for every vertex set $T \neq S$ and $T\neq \bar S$:
$$w(E(S, \bar S)\setminus E(T, \bar T)) > \gamma\cdot w(E(T, \bar T )\setminus E(S, \bar S)).$$
\end{Definition}

%%\paragraph{SDP Relaxation.} 
We use the Goemans--Williamson SDP relaxation for Max Cut with additional $\ell_2^2$-triangle inequalities~\cite{GW}. 
In the SDP relaxation, we have an SDP vector variable $\bar u$ for every vertex $u\in V$. The intended SDP solution
corresponding to an integral solution $(S, \bar S)$ assigns a fixed unit vector $\bar e$ to variable $\bar u$ if $u\in S$, and
assigns $-\bar e$ to $\bar u$, otherwise.
\begin{align*}
\text{maximize }& \frac{1}{4} \sum_{(u,v)\in E} w_{uv} \|\bar u - \bar v\|^2\\
\text{subject to:}& \ifSODA\text{ for every } u,v,w\in V\fi\\
& \|\bar u\|^2 = 1 \ifSODA\else\quad \text{for every } u\in V,\fi\\
&\|\bar u-\bar v\|^2 + \|\bar v-\bar w\|^2 \geq \|\bar u-\bar w\|^2 \ifSODA\else\quad
  \text{for every } u,v,w\in V,\fi\\
&\|\bar u-\bar v\|^2 + \|\bar v+\bar w\|^2 \geq \|\bar u+\bar w\|^2 \ifSODA\else\quad
  \text{for every } u,v,w\in V,\fi\\
&\|\bar u+\bar v\|^2 + \|\bar v+\bar w\|^2 \geq \|\bar u-\bar w\|^2 \ifSODA\else\quad
  \text{for every } u,v,w\in V\fi.
\end{align*}
The last three constraints are $\ell_2^2$-triangle inequalities for the set of vectors $\{\pm\bar u:u\in V\}$.
Note that the intended solution satisfies all constraints, and its value equals the cost of the cut~$(S,\bar S)$.

We show that $\gamma$-stable instances of Max Cut are integral. The formal definition of an integral SDP is as follows.

\begin{Definition} 
Let $G$ be an instance of Max Cut.
We say that an SDP solution $\{\bar u\}$ 
%of the SDP relaxation for $G$ 
is integral if  there exists a vector $\bar e$ such that $\bar u = \bar e$ or 
$\bar u = -\bar e$ for every $u\in V$.
We say that the SDP relaxation for $G$ is integral if every optimal SDP solution for $G$ is integral.
\end{Definition}
 
%%\paragraph{Robust Algorithm}
Our algorithm for $\gamma$-stable instances is robust in the sense of Raghavan and Spinrad~\cite{RS}:
%[RS] http://dl.acm.org/citation.cfm?id=365507
it always returns a correct output regardless of whether the input is $\gamma$-stable or not.
\begin{Definition} An algorithm for $\gamma$-stable instances of Max Cut is robust if the following conditions hold.
\begin{itemize}
\item If the input instance is $\gamma$-stable, the algorithm must output a maximum cut.
\item If the input instance is not $\gamma$-stable, the algorithm must either output a maximum cut or a special symbol $\perp$
(which certifies that the instance is not $\gamma$-stable).
\end{itemize}
\end{Definition}

\paragraph{Metrics of Negative Type and Sparsest Cut with Non-Uniform Demands.} In the proof, we use some standard definitions from metric geometry.
\begin{Definition}
The Lipschitz constant $\|\varphi\|_{Lip}$ of a map $\varphi$ between two metric spaces $(X,d_X)$ and $(Y, d_Y)$ 
equals 
$$\|\varphi\|_{Lip} = \sup_{\substack{x, y \in X \\ x\neq y}} \frac{d_Y(\varphi(x),\varphi(y))}{d_X(x,y)}.$$
The distortion of an embedding $\varphi:X \hookrightarrow Y$ equals $\|\varphi\|_{Lip} \cdot \|\varphi^{-1}\|_{Lip}$
(the distortion is infinite if $\varphi$ is not injective).
\end{Definition}

\begin{Definition} A set $X \subset \ell_2$ is an $\ell_2^2$ space, if the  distance function $d(u,v) = \|u-v\|^2$ 
satisfies triangle inequalities:
$$d(u,v) + d(v,w) \geq d(u,w) \quad\text{for every } u,v,w \in X.$$
A metric space is of negative type if it is isometric to an $\ell_2^2$ space.
We denote the least distortion with which every $n$ point metric space of negative type embeds into $\ell_1$ by $D_{\ell_2^2\to \ell_1} (n)$ or just $D_{\ell_2^2\to \ell_1}$.
\end{Definition}

Arora, Lee and Naor~\cite{ALN} proved that $D_{\ell_2^2\to \ell_1}(n) = O(\sqrt{\log n}\log\log n)$; on the other hand,
Cheeger, Kleiner and Naor~\cite{CKN} showed that $D_{\ell_2^2\to \ell_1}(n) = \Omega(\log^c n)$ for some constant $c>0$, and it is 
widely believed that $D_{\ell_2^2\to \ell_1}(n) = \Omega(\sqrt{\log n})$ (that is, that the upper bound of Arora, Lee and Naor is
almost optimal).

Note that since every $\ell_1$ metric can be represented as a linear combination of cut metrics with non-negative coefficients, we have the following fact.
%\begin{fact} 
\ifSODA
\begin{fact}{\sc (Linial, London and Rabinovich~\cite{LLR}, Aumann and Rabani~\cite{AR})}
\else
\begin{fact}[Linial, London and Rabinovich~\cite{LLR}, Aumann and Rabani~\cite{AR}]
\fi
\label{fact:cut-embedding}
Let $X \subset \ell_2$ be an $n$-point $\ell_2^2$ space ($n>1$). 
Then there exists a distribution of random cuts $(A, \bar A = X\setminus A)$
and a scale parameter $\sigma > 0$ such that for every pair $(x,y)\in X \times X$
\ifSODA
\begin{align*}
\sigma \cdot{} & \Prob{(x,y) \text{ is separated by } (A,\bar A)} \leq 
\|x-y\|^2 \leq \\
&\sigma \cdot D_{\ell_2^2\to \ell_1}(n) \cdot \Prob{(x,y) \text{ is separated by } (A,\bar A)}.
\end{align*}
\else
$$\sigma \cdot \Prob{(x,y) \text{ is separated by } (A,\bar A)} \leq \|x-y\|^2 \leq \sigma \cdot D_{\ell_2^2\to \ell_1}(n) \cdot \Prob{(x,y) \text{ is separated by } (A,\bar A)}.$$
\fi
\end{fact}

Now recall the definition the Sparsest Cut problem with non-uniform demands.
\ifSODA
\begin{Definition} {\sc (Sparsest Cut problem with Non-uniform Demands)}\label{def:nuSparseCut}
\else
\begin{Definition}[Sparsest Cut problem with Non-uniform Demands]\label{def:nuSparseCut}
\fi
We are given a connected graph $G=(V, E_c)$, 
with positive edge capacities $\mathrm{cap}_{uv}$, a set of demands pairs $E_d\subset V \times V$,
and positive demands $\mathrm{dem}_{uv}$. The sparsity of a cut $(A, \bar A)$ is
$$\phi(A) = \frac{\mathrm{cap}(E_c(A,\bar A))}{\mathrm{dem}(E_d(A,\bar A))},
$$
where $E_c(A,\bar A)$ is the set of capacity edges between $A$ and $\bar A$, $\mathrm{cap}_{uv}(E_c(A,\bar A))$ is their total capacity,
$E_d(A,\bar A)$ is the set of demand pairs separated by $(A,\bar A)$, and $\mathrm{dem}(E_d(A,\bar A))$ is their total demand.
Our goal is to find a cut $(A,\bar A)$ with minimum sparsity $\phi(A)$.

We denote the best possible approximation factor for the problem by $\alpha_{SC}(n)$. 
%That is, there is no polynomial-time algorithm that has approximation ratio better than $\alpha_{SC}(n)$ for every $n$. 
Strictly speaking, $\alpha_{SC}(n)$ is not well-defined. 
Formally, we consider the decision version of the problem: the approximation 
algorithm has to output only the approximate value of the problem. 
We write $\alpha_{SC}(n) \leq f(n)$ if there is an algorithm with approximation guarantee $f(n)$; we write $\alpha_{SC}(n) > f(n)$ if there is no algorithm with approximation guarantee $f(n)$. However, the algorithm of Arora, Lee, and Naor~\cite{ALN} that we use in this paper not only finds the approximate value but also finds the corresponding solution.
%Theorems we prove below hold with any such $\alpha_{SC}(n)$.)
\end{Definition}

Arora, Lee, and Naor showed that $\alpha_{SC}(n) \leq D_{\ell_2^2\to \ell_1}(n) = O(\sqrt{\log n} \log\log n)$.
Chawla, Krauthgamer, Kumar, Rabani, and Sivakumar~\cite{CKKRS} and independently Khot and Vishnoi~\cite{KV}
proved that $\alpha_{SC}(n)\to\infty$ as $n\to \infty$ assuming the Unique Games Conjecture (UGC); moreover, they showed that $\alpha_{SC}(n) = \Omega((\log\log n)^{1/2})$ assuming some strong version of UGC.
Chuzhoy and Khanna~\cite{CK06} proved that there is no polynomial-time approximation scheme for Sparsest Cut; that is, that $\alpha_{SC} > 1 + \varepsilon$
for some absolute constant $\varepsilon > 0$.
%Recently, Amb\"uhl, Mastrolilli, and Svensson\cite{AMS} proved that $\alpha_{SC}(n) > 1$ unless 
%SAT can be solved in randomized subexponential time.

\section{Algorithm for Max Cut} \label{sec:robustalgo}
In this section, we prove that the SDP relaxation for every $\gamma$-stable Max Cut instance with $\gamma \geq D_{\ell_2^2\to \ell_1}$
is integral and as a corollary obtain a polynomial-time algorithm for $\gamma$-stable instances.
  
\begin{theorem}\label{thm:sdp-integral}
The SDP relaxation for every $\gamma$-stable Max Cut instance with $\gamma \geq D_{\ell_2^2\to \ell_1}$
is integral.
\end{theorem}
\begin{proof}
Let $G$ be a $\gamma$-stable instance with $\gamma \geq D_{\ell_2^2\to \ell_1}$ and $\{\bar u\}$ be an optimal SDP solution 
for $G$. Assume to the contrary that $\{\bar u\}$ is not integral. 

Let $(S,\bar S)$ be the maximum cut in $G$. Since $\{\bar u\}$ is an optimal SDP solution, we have that its SDP value 
is at least the cost of the maximum cut:
\begin{equation}\label{ineq:sdp-is-relaxation}
\frac{1}{4} \sum_{(u,v)\in E} w_{uv} \|\bar u - \bar v\|^2 \geq w(E(S,\bar S)).
\end{equation}
Define
$$\hat u=
\begin{cases}
\bar u, &\text{if } u\in S,\\
-\bar u, &\text{if } u\notin S.
\end{cases}
$$
Note that not all vectors $\hat u$ are equal since the SDP solution is not integral.
If $(u,v) \in E(S,\bar S)$ then 
$$\|\bar u - \bar v\|^2 = 2 - 2 \langle \bar u , \bar v\rangle = 2 + 2 \langle \hat u , \hat v \rangle = 4 - \|\hat u - \hat v\|^2;$$
if $(u,v) \in E \setminus E(S,\bar S)$ then $\|\bar u - \bar v\|^2 = \|\hat u - \hat v\|^2$.
Therefore,
\ifSODA
\begin{align*}
\frac{1}{4} & \sum_{(u,v)\in E} w_{uv} \|\bar u - \bar v\|^2 = 
\frac{1}{4} \sum_{(u,v)\in E(S,\bar S)} w_{uv} \|\bar u - \bar v\|^2 \\
&\qquad{}+ \frac{1}{4} \sum_{(u,v)\in E\setminus E(S,\bar S)} w_{uv} \|\bar u - \bar v\|^2\\
&= \frac{1}{4} \sum_{(u,v)\in E(S,\bar S)} w_{uv} (4 - \|\hat u - \hat v\|^2) \\
&\qquad{}+ \frac{1}{4} \sum_{(u,v)\in E\setminus E(S,\bar S)} w_{uv} \|\hat u - \hat v\|^2\\
&= w(E(S,\bar S)) + \frac{1}{4}  \sum_{(u,v)\in E \setminus E(S,\bar S)} w_{uv} \|\hat u - \hat v\|^2\\
&\qquad   {}- \frac{1}{4}  \sum_{(u,v)\in E(S,\bar S)} w_{uv} \|\hat u - \hat v\|^2.
\end{align*}
\else
\begin{align*}
\frac{1}{4} \sum_{(u,v)\in E} w_{uv} \|\bar u - \bar v\|^2 &= 
\frac{1}{4} \sum_{(u,v)\in E(S,\bar S)} w_{uv} \|\bar u - \bar v\|^2 + \frac{1}{4} \sum_{(u,v)\in E\setminus E(S,\bar S)} w_{uv} \|\bar u - \bar v\|^2\\
&= \frac{1}{4} \sum_{(u,v)\in E(S,\bar S)} w_{uv} (4 - \|\hat u - \hat v\|^2) + \frac{1}{4} \sum_{(u,v)\in E\setminus E(S,\bar S)} w_{uv} \|\hat u - \hat v\|^2\\
&= w(E(S,\bar S)) + \frac{1}{4}  \sum_{(u,v)\in E \setminus E(S,\bar S)} w_{uv} \|\hat u - \hat v\|^2
   - \frac{1}{4}  \sum_{(u,v)\in E(S,\bar S)} w_{uv} \|\hat u - \hat v\|^2.
\end{align*}
\fi
From (\ref{ineq:sdp-is-relaxation}), we get
$$
\sum_{(u,v)\in E(S,\bar S)} w_{uv} \|\hat u - \hat v\|^2
\leq
\sum_{(u,v)\in E \setminus E(S,\bar S)} w_{uv} \|\hat u - \hat v\|^2
.$$
The set $X= \{\hat u: u\in V\}$ is an $\ell_2^2$ space since the SDP solution ensures that vectors in
$\set{\pm\bar u: u\in V}$ satisfy $\ell_2^2$-triangle inequalities and $X \subset \set{\pm\bar u: u\in V}$.
This space embeds into $\ell_1$ with distortion $D_{\ell_2^2\to \ell_1}(n)$, and, hence (see Fact~\ref{fact:cut-embedding})
there is a distribution of random cuts $(A,\bar A)$ of $X$ and a parameter $\sigma > 0$ such that
\ifSODA
\begin{align*}
\sigma & \Prob{\text{pair }(\hat u,\hat v) \text{ is separated by } (A,\bar A)} \leq \|\hat u -\hat v\|^2 \\
&\leq
\sigma  D_{\ell_2^2\to \ell_1}(n)  \Prob{\text{pair }(\hat u,\hat v) \text{ is separated by } (A,\bar A)}.
\end{align*}
\else
\begin{multline*}
\sigma \cdot \Prob{\text{pair }(\hat u,\hat v) \text{ is separated by } (A,\bar A)} \leq \|\hat u -\hat v\|^2 \\
\leq
\sigma \cdot D_{\ell_2^2\to \ell_1}(n) \cdot \Prob{\text{pair }(\hat u,\hat v) \text{ is separated by } (A,\bar A)}.
\end{multline*}
\fi
Here, we use that not all vectors in $X$ are equal and therefore cuts in the distribution are not trivial.
Let $A' = \{u:\hat u \in A\}$ and $\bar A' = V\setminus A' = \{u:\hat u \notin A\}$. We get, 
\ifSODA
\begin{align*}
&\sum_{(u,v)\in E(S,\bar S)} w_{uv} \sigma \Prob{(u, v) \in E(A',\bar A')} \\
&{}\leq
\sum_{(u,v)\in E(S,\bar S)} w_{uv} \|\hat u - \hat v\|^2
\leq
\sum_{(u,v)\in E\setminus E(S,\bar S)} w_{uv} \|\hat u - \hat v\|^2 \\
&{}\leq \sum_{(u,v)\in E\setminus E(S,\bar S)} w_{uv} \sigma D_{\ell_2^2\to \ell_1}(n) \Prob{(u, v) \in (A',\bar A')}.
\end{align*}
\else
\begin{align*}
\sum_{(u,v)\in E(S,\bar S)} w_{uv} & \sigma \Prob{(u, v) \in E(A',\bar A')} 
\leq
\sum_{(u,v)\in E(S,\bar S)} w_{uv} \|\hat u - \hat v\|^2\\
&\leq
\sum_{(u,v)\in E\setminus E(S,\bar S)} w_{uv} \|\hat u - \hat v\|^2 \\
&\leq \sum_{(u,v)\in E\setminus E(S,\bar S)} w_{uv} \sigma D_{\ell_2^2\to \ell_1}(n) \Prob{(u, v) \in (A',\bar A')}.
\end{align*}
\fi
Therefore,
\ifSODA
\begin{multline*}
\E{w(E(S,\bar S) \cap E(A',\bar A'))} \\ \leq D_{\ell_2^2\to \ell_1}(n) \cdot \E{w((E\setminus E(S,\bar S)) \cap E(A',\bar A'))}.
\end{multline*}
\else
$$\E{w(E(S,\bar S) \cap E(A',\bar A'))} \leq D_{\ell_2^2\to \ell_1}(n) \cdot \E{w((E\setminus E(S,\bar S)) \cap E(A',\bar A'))}.$$
\fi
In particular, for some cut $A''$, we have 
\ifSODA
\begin{align*}
w&(E(S,\bar S) \cap E(A'',\bar A'')) \\
&\leq D_{\ell_2^2\to \ell_1}(n) \cdot w((E\setminus E(S,\bar S)) \cap E(A'',\bar A''))\\
&\leq \gamma \cdot w((E\setminus E(S,\bar S)) \cap E(A'',\bar A'')).
\end{align*}
\else
\begin{align*}
w(E(S,\bar S) \cap E(A'',\bar A'')) &\leq D_{\ell_2^2\to \ell_1}(n) \cdot w((E\setminus E(S,\bar S)) \cap E(A'',\bar A''))\\
&\leq \gamma \cdot w((E\setminus E(S,\bar S)) \cap E(A'',\bar A'')).
\end{align*}
\fi
Let $T = (S \cap A'') \cup (\bar S \cap \bar A'')$ (see Figure~\ref{fig:cuts}). Note that $A'' \neq V$ and $A''\neq \varnothing$, hence
$T\neq S$ and $T\neq \bar S$. Write 
\ifSODA
\begin{align*}
E(S,\bar S)  & {}\cap E(A'',\bar A'') \\
&= E(S\cap A'', \bar S \cap \bar A'') \cup E(S\cap \bar A'', \bar S \cap A'')\\
&=E(S\cap T, \bar S \cap T) \cup E(S\cap \bar T, \bar S \cap \bar T)\\
&= E(S, \bar S) \setminus E(T, \bar T),\\
\intertext{and}
\mathrlap{(E\setminus E(S,\bar S)) \cap E(A'',\bar A'')} \phantom{E(S,\bar S)} &\\
&= E(S\cap A'', S\cap \bar A'') \cup E(\bar S\cap A'', \bar S\cap \bar A'')\\
&= E(T\cap A'', \bar T\cap A'') \cup E(T\cap \bar A'', \bar T\cap \bar A'') \\
&= E(T, \bar T) \cap E(S, \bar S).
\end{align*}
\else
\begin{align*}
E(S,\bar S) \cap E(A'',\bar A'')&= E(S\cap A'', \bar S \cap \bar A'') \cup E(S\cap \bar A'', \bar S \cap A'')\\
&=E(S\cap T, \bar S \cap T) \cup E(S\cap \bar T, \bar S \cap \bar T)=
E(S, \bar S) \setminus E(T, \bar T),\\
\intertext{and}
(E\setminus E(S,\bar S)) \cap E(A'',\bar A'') &= E(S\cap A'', S\cap \bar A'') \cup E(\bar S\cap A'', \bar S\cap \bar A'')\\
&= E(T\cap A'', \bar T\cap A'') \cup E(T\cap \bar A'', \bar T\cap \bar A'')\\
& = E(T, \bar T) \cap E(S, \bar S).
\end{align*}
\fi
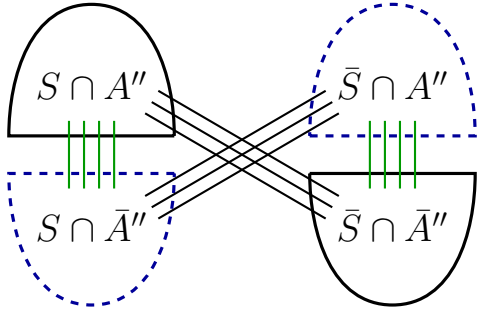
\begin{figure}
%\begin{center}
\centering
\begin{tikzpicture}
    \draw[very thick] (-0.9,0.25) to (-3.1, 0.25) to [out=90,in=180] (-2, 2) to [out=0,in=90] (-0.9, 0.25);
    \draw[very thick] (0.9,-0.25) to (3.1, -0.25) to [out=-90,in=0] (2, -2) to [out=180,in=-90] (0.9, -0.25);

    \draw[very thick, dashed, color=black!40!blue] (+0.9,0.25) to (+3.1, 0.25) to [out=90,in=0] (+2, 2) to [out=180,in=90] (+0.9, 0.25);
    \draw[very thick, dashed, color=black!40!blue] (-0.9,-0.25) to (-3.1, -0.25) to [out=-90,in=180] (-2, -2) to [out=0,in=-90] (-0.9, -0.25);

    \foreach \x in {-1.7,-1.9,-2.1, -2.3, 1.7, 1.9, 2.1, 2.3}
       \draw[thick, color=black!40!green] (\x, 0.45) -- (\x,-0.45);
    
    \draw[thick, color=black] (-1.2+0.0875, 0.85) -- (1.2+0.0875, -0.55);
    \draw[thick, color=black] (-1.2, 0.7) -- (1.2, -0.7);
    \draw[thick, color=black] (-1.2-0.0875, 0.55) -- (1.2-0.0875, -0.85);

    \draw[thick, color=black] (-1.2+0.0875, -0.85) -- (1.2+0.0875, 0.55);
    \draw[thick, color=black] (-1.2, -0.7) -- (1.2, 0.7);
    \draw[thick, color=black] (-1.2-0.0875, -0.55) -- (1.2-0.0875, 0.85);
    \node[above] at (-2,0.6) {\Large $S\cap A''$};
    \node[above] at (2,0.6) {\Large $\bar S\cap A''$};
    \node[below] at (-2,-0.6) {\Large $S\cap \bar A''$};
    \node[below] at (2,-0.6) {\Large $\bar S\cap \bar A''$};
\end{tikzpicture}
%
%\includegraphics[scale=0.75]{figure-1-crop.pdf}
%\end{center}
\caption{In the figure, the set $T = (S \cap A'') \cup (\bar S \cap \bar A'')$ is represented
by semicircles in the upper-left and lower-right corners; the set $\bar T = (S \cap \bar A'') \cup (\bar S \cap A'')$ 
is represented by semicircles in the lower-left and upper-right corners. Edges in $E(S, \bar S) \setminus E(T, \bar T)$
are drawn diagonally; edges in $E(T, \bar T) \setminus E(S, \bar S)$ are drawn vertically.}
\label{fig:cuts}
\end{figure}
We get,
$$w(E(S,\bar S) \setminus E(T,\bar T)) \leq \gamma\cdot w(E(T,\bar T) \setminus E(S,\bar S)),$$
which contradicts to the fact that $G$ is a $\gamma$-stable instance (see Definition~\ref{def:stability2}).
\end{proof}

From Theorem~\ref{thm:sdp-integral}, we get the main algorithmic result of our paper.
\begin{corollary}\label{cor:algorithm}
There is a robust polynomial-time algorithm for $\gamma$-stable instance of Max Cut with $\gamma \geq D_{\ell_2^2\to \ell_1}(n)$. 
%$D_{\ell_2^2\to \ell_1}(n)= O(\sqrt{\log n}\log\log n)$.
\end{corollary}
\begin{proof}
The algorithm solves the SDP relaxation for the problem. If the solution is integral, the algorithm returns the cut
corresponding to it. Otherwise, it returns $\perp$ (indicating that the instance is not $\gamma$-stable). Note that if the algorithm
returns a cut, it must be a maximum cut (otherwise, the SDP solution would not be optimal).
By Theorem~\ref{thm:sdp-integral}, the algorithm always returns a solution if the instance is $\gamma$-stable.
\end{proof}

\section{Algorithm for Minimum Multiway Cut}\label{sec:muliway}
In this section, we study stable instances of Minimum Multiway Cut. We prove that the linear programming relaxation 
of C\u{a}linescu, Karloff, and Rabani~\cite{CKR} is integral for $4$-stable instances of the problem. Thus there is a
robust polynomial-time algorithm for 4-stable instances of Minimum Multiway Cut.

The Minimum Multiway Cut problem was introduced by 
Dahlhaus, Johnson, Papadimitriou, Seymour, and Yannakakis~\cite{DJPSY}.
We refer the reader to \cite{BNS} for the summary of known results for the problem.
\begin{Definition}
In the Minimum Multiway Cut problem, we are given a graph $G=(V,E,w)$ with positive edge weights $w_e$ and a set of terminals $T=\{s_1,\dots, s_k\}\subset V$.
Our goal is to partition the graph into $k$ pieces $S_1,\dots, S_k$ such that $s_i\in S_i$ so as 
to minimize the total weight of cut edges.
\end{Definition}

We give a definition of $\gamma$-stable instances of Minimum Multiway Cut (cf. Definition~\ref{def:stability1}).
\begin{Definition}
Let $\gamma > 1$.  An instance $\{G=(V,E,w),T\}$ of Minimum Multiway Cut is $\gamma$-stable if there is a multiway cut ${\cal S}$
which is the unique optimal solution for every $\gamma$-perturbation of $G$.
\end{Definition}

We also restate this definition as follows (cf. Definition~\ref{def:stability2}).

\begin{Definition}
Consider an instance $\{G=(V,E,w),\, T\}$ of Minimum Multiway Cut. Let $\gamma > 1$. Denote the optimal multiway cut by ${\cal S}^*$, and let
the set it cuts be $E^*$. We say that $G$ is a $\gamma$-stable instance of Multiway Cut if for every multiway cut ${\cal S}'\neq {\cal S}^*$, we have
$$w(E'\setminus E^*) > \gamma\cdot w(E^*\setminus E'),$$
where $E'$ is the set of edges cut by ${\cal S}'$.
\end{Definition}

Consider the LP relaxation of C\u{a}linescu, Karloff, and Rabani~\cite{CKR}.
In this relaxation, there is a variable $\bar u=(u_1,\dots, u_k)\in {\mathbb R}^k$
for every vertex $u\in V$. Let $e_1,\dots, e_k$ be the standard basis in ${\mathbb R}^k$ and $\Delta = \{x:\|x\|_1 = 1, x_1\geq 0,\dots, x_k\geq 0\}$
be the simplex with vertices $e_1,\dots, e_k$.
\ifSODA
\begin{align}
\text{minimize } & \mathrlap{\frac{1}{2} \sum_{(u,v)\in E} w(u,v) \, \|\bar u - \bar v\|_1}\qquad \label{eq:CKR-relaxation}\\
\text{subject to: }&\notag\\
&\bar s_i = e_i &&\text{for every } i,\notag\\
&\bar u \in \Delta &&\text{for every } u\in V.\notag
\end{align}
\else
\begin{align}
\text{minimize } & \frac{1}{2} \sum_{(u,v)\in E} w(u,v) \, \|\bar u - \bar v\|_1 \label{eq:CKR-relaxation}\\
\text{subject to: }&\notag\\
&\bar s_i = e_i &&\text{for every } i,\notag\\
&\bar u \in \Delta &&\text{for every } u\in V.\notag
\end{align}
\fi
Every feasible LP solution defines a metric on $V$: $d(u,v) = \|\bar u - \bar v\|_1/2$.
We will need the following lemma.
\begin{lemma}\label{lem:hol-rounding}
Consider a feasible LP solution $\{\bar u:u\in V\}$. There is a distribution of multiway cuts (partitions) $S_1,\dots, S_k$ such that
\begin{itemize}
\item $s_i\in S_i$ for every $i\in \{1,\dots, k\}$ (always),
\item $\Pr(u \text{ and } v \text{ are separated by the cut}) \leq \frac{2d(u,v)}{1+d(u,v)}$ for every $u$ and $v$
($u$ and $v$ are separated if $u\in S_i$ and $v\in S_j$ with $i\neq j$).
 In particular, for every edge $(u,v)$ (see Figure~\ref{fig:graph-prob-cut}),
\ifSODA
$$
\Pr((u, v) \text{ is cut}) \leq 2d(u,v)$$ and $$\Pr((u,v)  \text{ is not cut}) \geq \frac{1 - d(u,v)}{2}.$$
\else
$$
\Pr((u, v) \text{ is cut}) \leq 2d(u,v) \quad\text{ and }\quad
\Pr((u,v)  \text{ is not cut}) \geq \frac{1 - d(u,v)}{2}.
$$
\fi
\end{itemize}
If the LP solution is not integral, the distribution is supported on at least two multiway cuts.
\end{lemma}
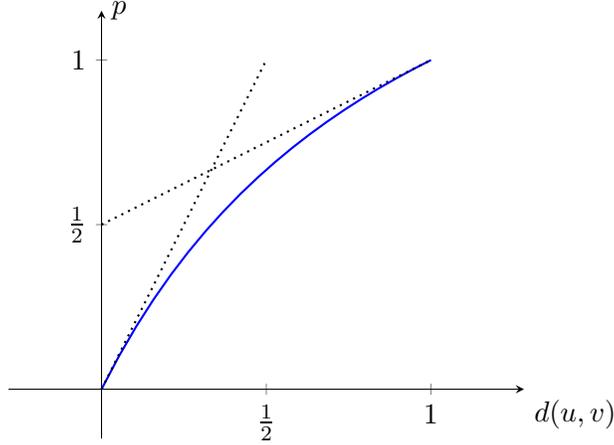
\begin{figure}
%\begin{center}
\centering
\begin{tikzpicture}
\begin{axis}[
axis x line=middle,
axis y line=middle,
xmax=1, xmin=0,
ymax=1, ymin=0,
xlabel={$d(u,v)$}, ylabel={$p$},
%ylabel={$\frac{2d(u,v)}{1+d(u,v)}$},$\Pr(u \text{ and } v \text{ are separated})$,
axis equal = true,
xtick={0,0.5,1},
ytick={0,0.5,1},
xticklabels={$0$,$\frac12$,$1$},
yticklabels={$0$,$\frac12$,$1$},
enlargelimits=0.15,
%scale=0.6, %ifSODA
xlabel style={at={(current axis.right of origin)},anchor=north west},
ylabel style={at={(current axis.above origin)},anchor=west}
]
\addplot[blue,mark=none, thick, domain=0:1,samples=20]{2*x/(1+x)};
\addplot[black,dotted,thick, domain=0:0.5,samples=3]{2*x};
\addplot[black,dotted, thick,domain=0:1,samples=3]{1/2 + x/2};

\end{axis}
\end{tikzpicture}
%\end{center}
\caption{The probability that $u$ and $v$ are separated is at most $p = 2d(u,v) / (1 + d(u,v))$, which is bounded from above by $2d(u,v)$ and by $1 - (1-d(u,v))/2$.}
\label{fig:graph-prob-cut}
\end{figure}
\begin{proof}
We use the rounding scheme of Kleinberg and Tardos~\cite{KT02} (which they used in their algorithm for the Metric Labeling problem)
to round the LP solution to an integral solution. The scheme works as follows.
We iteratively construct sets $S_1,\dots, S_k$. We start with empty sets
$S_1,\dots, S_k$ and then in each iteration add vertices to one of the sets $S_1,\dots, S_k$. We stop once
each vertex $u$ belongs to some set $u$. In each iteration, we choose independently and uniformly at random $r\in (0,1)$ and $i\in \{1,\dots, k\}$.
We add each vertex $u$ to $S_i$ if $r\leq \bar u_i$ and it was not added to any set $S_j$ in previous iterations.

First, note that we add every vertex $u$ to some $S_i$ with probability $\sum_{i=1}^k \bar u_i/k = 1/k$ in each iteration
(unless $u$ already lies in some $S_j$). So eventually we will add every vertex to some set $S_i$.
Also note that we cannot add $s_i$ to $S_j$ if $j\neq i$. Therefore, $s_i\in S_i$.

Now consider two vertices $u$ and $v$. Consider an iteration of our partitioning algorithm. Suppose that neither $u$ nor $v$ is assigned
to some $S_j$. The probability that at least one of them is assigned to some $S_i$ in this iteration is
\ifSODA
\begin{align*}
\frac{1}{k}\sum_{i=1}^k \max(\bar u_i,\bar v_i) &= \frac{1}{k}\sum_{i=1}^k \left(\frac{\bar u_i +\bar v_i}{2} 
+ \frac{|\bar u_i -\bar v_i|}{2}\right) \\
&= \frac{1}{k}\left(1 
+ \frac{\|\bar u -\bar v\|_1}{2}\right) \\
&= \frac{1 
+ d(u,v)}{k}.
\end{align*}
\else
$$\frac{1}{k}\sum_{i=1}^k \max(\bar u_i,\bar v_i) = \frac{1}{k}\sum_{i=1}^k \left(\frac{\bar u_i +\bar v_i}{2} 
+ \frac{|\bar u_i -\bar v_i|}{2}\right) = \frac{1}{k}\left(1 
+ \frac{\|\bar u -\bar v\|_1}{2}\right) = \frac{1 
+ d(u,v)}{k}.$$
\fi
The probability that exactly one of them is assigned to some $S_i$ is
$$\frac{1}{k}\sum_{i=1}^k |\bar u_i-\bar v_i| = \frac{\|\bar u - \bar v\|_1 }{k}= \frac{2d(u,v)}{k}.$$
Therefore, the probability that $u$ and $v$ are separated in some iteration is $2d(u,v) / (1 + d(u,v))$.
The probability that $u$ and $v$ belong to different pieces of the cut is at most $2d(u,v) / (1 + d(u,v))$.

Finally, note that if some $\bar u_j \in (0,1)$ then with positive probability $u\in S_j$,
and with positive probability $u\notin S_j$. Therefore, if the LP solution is not integral, the distribution of multiway cuts is supported on at least two multiway cuts.
\end{proof}
\begin{remark}
Note that in general this rounding scheme gives only a $2$ approximation for Multiway Cut. 
Other known rounding schemes achieve a better approximation; e.g.
the rounding scheme of C\u{a}linescu, Karloff, and Rabani~\cite{CKR} gives a $3/2$ approximation.
However, this rounding scheme has a property that other rounding schemes do not have: it does not cut an edge $(u,v)$ 
with probability at least $(1 - d(u,v))/{2} = \Omega(1 - d(u,v))$. This property is crucial for our proof (we discuss why this property is important
in Section~\ref{sec:Discussion}).
\end{remark}

Now we prove the main result of this section.
\begin{theorem}\label{thm:main-multiway}
The LP relaxation is integral if the instance is $4$-stable.
\end{theorem}
\begin{proof}
Consider a $4$-stable instance of Multiway Cut. 
%Let $S_1, \dots, S_k$ be its optimal solution.
Let ${\cal S}^*$ be the minimum multiway cut and $E^*$ be the set of edges cut by it.
Let $\{\bar u:u\in V\}$ be the optimal solution of the LP relaxation. 
Assume to the contrary that it is not integral.
Consider a random multiway ${\cal S}'=(S_1,\dots, S_k)$ as in Lemma~\ref{lem:hol-rounding}.
Let $E'$ be the set of edges cut by it ($E'$ is a random variable).
Note that since the input instance is $4$-stable, we have
$4w(E^*\setminus E') < w(E' \setminus E^*)$
unless ${\cal S}' = {\cal S}^*$. 
Since the LP solution is not integral,
${\cal S}' \neq {\cal S}^*$ with positive probability, and
thus $4\E{w(E^*\setminus E')} < \E{w(E' \setminus E^*)}$.

Let 
\begin{align*}
\mathsf{LP}_{+} &= \sum_{(u,v)\in E^*} w(u,v) (1 - d(u,v)),\\
\mathsf{LP}_{-} &= \sum_{(u,v)\in E\setminus E^*} w(u,v) \, d(u,v).
\end{align*}
We have, 
\ifSODA
\begin{align*}
\E{w(E^*\setminus E')} &= \sum_{(u,v) \in E^*} w(u,v) \Pr((u,v) \text{ is not cut}) \\
&
\geq 
\sum_{(u,v) \in E^*} w(u,v) (1-d(u,v))/2 \\&
= \mathsf{LP}_{+}/2,\\
\intertext{and}
\E{w(E'\setminus E^*)} &= \sum_{(u,v) \in E\setminus E^*} w(u,v) \, \Pr((u, v) \text{ is cut}) \\
&\leq 
2\sum_{(u,v) \in E^*} w(u,v) d(u,v) \\
&= 2\,\mathsf{LP}_{-}.
\end{align*}
\else
\begin{align*}
\E{w(E^*\setminus E')} &= \sum_{(u,v) \in E^*} w(u,v) \Pr((u,v) \text{ is not cut}) \\
&
\geq 
\sum_{(u,v) \in E^*} w(u,v) (1-d(u,v))/2 = \mathsf{LP}_{+}/2,\\
\E{w(E'\setminus E^*)} &= \sum_{(u,v) \in E\setminus E^*} w(u,v) \, \Pr((u, v) \text{ is cut}) \leq 
2\sum_{(u,v) \in E^*} w(u,v) d(u,v) = 2\,\mathsf{LP}_{-}.
\end{align*}
\fi
We conclude that $\mathsf{LP}_{+} < \mathsf{LP}_{-}$.
On the other hand,
\ifSODA
\begin{align*}
\mathsf{LP}_{+} - \mathsf{LP}_{-} &= w(E^*) - \sum_{(u,v) \in E} w(u,v) \, d(u,v) \\
&= w(E^*) - \frac{1}{2}\sum_{(u,v) \in E} w(u,v) \, \|\bar u-\bar v\|_1 
\geq 0
\end{align*}
\else
$$\mathsf{LP}_{+} - \mathsf{LP}_{-} = w(E^*) - \sum_{(u,v) \in E} w(u,v) \, d(u,v) = w(E^*) - \frac{1}{2}\sum_{(u,v) \in E} w(u,v) \, \|\bar u-\bar v\|_1 \geq 0$$
\fi
since the value of the relaxation is at most the value of the integral solution. We get a contradiction.
\end{proof}

As an immediate corollary we get that there is a robust polynomial-time algorithm for $4$-stable instances of Multiway Cut.
\begin{corollary}
There is a robust polynomial-time algorithm for $4$-stable instances of Multiway Cut.
\end{corollary}
\begin{proof}
We solve the LP relaxation for Multiway Cut. If the LP solution is integral, we return the corresponding combinatorial solution. Otherwise,
we return that the instance is not $4$-stable.
\end{proof}

\section{Negative Results}\label{sec:negative}
In this section, we present our hardness results and prove that the SDP relaxation is not integral if $\gamma < D_{\ell_2^2\to \ell_1}(n/2)$.

\subsection{Reduction from Sparsest Cut to Max Cut} \label{sec:reduction}
We first present a reduction from Sparsest Cut to Max Cut, which we use later to prove both our negative results.

Consider a Sparsest Cut instance. Denote the set of vertices by $V_0$, the set of capacity edges by $E_c$, the set of demand pairs by $E_d$, 
edge capacities by $\mathrm{cap}_{uv}$, and demands by $\mathrm{dem}_{uv}$. Define graph $G(V, E, w)$ as follows. 
Introduce two vertices $u_1$ and $u_2$ for every $u\in V_0$, and let $V=\set{u_1,u_2:u\in V_0}$. Let 
\ifSODA
\begin{align*}
E &= \set{(u_1,v_2): (u,v) \in E_c} \\ &\qquad{} \cup  \set{(u_1,v_1), (u_2,v_2): (u,v) \in E_d} \\&\qquad{}\cup \set{(u_1,u_2): u\in V_0}.
\end{align*}
\else
$$E = \set{(u_1,v_2): (u,v) \in E_c} \cup \set{(u_1,v_1), (u_2,v_2): (u,v) \in E_d} \cup \set{(u_1,u_2): u\in V_0}.$$
\fi
Define edge weights by $w(u_1,v_2) = \mathrm{cap}_{uv}$, $w(u_1,v_1) = w(u_2, v_2) = \mathrm{dem}_{uv}$, and $w(u_1, u_2) = W_{\infty}$,
where $W_{\infty}$ is an arbitrary number larger than 
\ifSODA
$\gamma \cdot w(\set{(u_1,v_2): (u,v) \in E_c} \cup \set{(u_1,v_1), (u_2,v_2): (u,v) \in E_d})$.
\else
$$\gamma \cdot w(\set{(u_1,v_2): (u,v) \in E_c} \cup \set{(u_1,v_1), (u_2,v_2): (u,v) \in E_d}).$$
\fi
%(alternatively, $W_{\infty} = +\infty$).
Let $S = \set{u_1:u\in V_0}$ and $\bar S = V\setminus S = \set{u_2: u\in V_0}$. 

\begin{lemma}\label{lem:reduction}
If $\phi(A) > \gamma$  for every cut $(A,\bar A)$ (see Definition~\ref{def:nuSparseCut}), then
the instance $G$ is $\gamma$-stable with the maximum cut $(S,\bar S)$.
\end{lemma}
\begin{proof}
We need to show that for every cut $(T,\bar T)$ different from $(S,\bar S)$:
$$w(E(S,\bar S)\setminus E(T,\bar T)) > \gamma \cdot w(E(T,\bar T)\setminus E(S,\bar S)).$$
Here, we use Definition~\ref{def:stability2} of $\gamma$-stability. Note that if for some $u$, the edge $(u_1,u_2)$ is not cut by $E(T,\bar T)$ then 
$w(E(S,\bar S)\setminus E(T,\bar T)) \geq w(u_1,u_2) = W_{\infty}$ and $\gamma \cdot w(E(T,\bar T)\setminus E(S,\bar S)) < W_{\infty}$,
and the desired inequality holds. So we assume below that every edge $(u_1,u_2)$ is cut by $E(T,\bar T)$.
Then, for every $u$ either $u_1\in T$ and $u_2\in \bar T$, or $u_1\in \bar T$ and $u_2\in T$.
Let 
\begin{align}
A &= \set{u\in V_0: u_1\in T} = \set{u\in V_0: u_2\in \bar T} ;\\
\bar A &= \set{u\in V_0: u_1\in \bar T} = \set{u\in V_0: u_2\in T}.
\end{align}
Observe, that $S\cap T = \{u_1: u \in A\}$; $S\cap \bar T = \{u_1: u \in \bar A\}$, similarly
$\bar S\cap T = \{u_2: u \in \bar A\}$; $\bar S\cap \bar T = \{u_2: u \in A\}$. 
Since $\phi(A) > \gamma$, we have
\ifSODA
\begin{align*}
\mathrlap{w(E(S,\bar S)\setminus E(T,\bar T)) }\qquad & \\
&=w(E(S\cap T, \bar S\cap T)) + w(E(S\cap \bar T, \bar S\cap \bar T))\\
&= 2\, \mathrm{cap}(E_c(A,\bar A)) > \gamma \cdot 2\, \mathrm{dem}(E_d(A,\bar A))\\
&= w(E(S\cap T, S\cap \bar T)) + w(E(\bar S\cap T, \bar S\cap \bar T))\\
& = w(E(T,\bar T)\setminus E(S,\bar S)),
\end{align*}
\else
\begin{align*}
w(E(S,\bar S)\setminus E(T,\bar T)) &= w(E(S\cap T, \bar S\cap T)) + w(E(S\cap \bar T, \bar S\cap \bar T))\\
&= 2\, \mathrm{cap}(E_c(A,\bar A)) > \gamma \cdot 2\, \mathrm{dem}(E_d(A,\bar A))\\
&= w(E(S\cap T, S\cap \bar T)) + w(E(\bar S\cap T, \bar S\cap \bar T)) = w(E(T,\bar T)\setminus E(S,\bar S)),
\end{align*}
\fi
as required. We proved that the instance is $\gamma$-stable.
\end{proof}

\subsection{Hardness Result for Max Cut} \label{sec:hardness}
We now prove that there is no robust polynomial-time algorithm for $\gamma$-stable instances of Max Cut 
when $\gamma < \alpha_{SC}(n/2)$.

\begin{theorem}\label{thm:hardness}
Suppose that there is a robust polynomial-time algorithm $\cal A$ for $\gamma$-stable instances of Max Cut with $\gamma \geq \gamma(n)$.
Then there is a polynomial-time algorithm $\cal B$ for the decision version of Sparsest Cut with promise that either
$$\phi^* = \min_{(A,\bar A)} \phi(A) < \phi_0 \text{ or } \phi^* > \gamma(2n) \phi_0.$$
The algorithm given a Sparsest Cut instance decides whether $\phi^* < \phi_0$ or $\phi^* > \gamma(2n) \phi_0$.
\end{theorem}
\begin{proof}
We may assume that $\phi_0 = 1$ by dividing all edge weights by $\phi_0$. We apply reduction from Section~\ref{sec:reduction} and obtain a graph $G$ on $2n$ vertices. Then we
run $\cal A$ on $G$. If $\cal A$ returns the cut $(S, \bar S)$ (where $S = \{u_1:u\in V_0\}$),
we decide that $\phi^* > \gamma(2n)$. Otherwise, we decide that $\phi^* < 1$.

We prove that we always decide correctly. Assume first that $\phi^* > \gamma(2n)$ then $G$ is $\gamma(2n)$-stable and $(S, \bar S)$ is the maximum cut
by Lemma~\ref{lem:reduction}. Therefore, $\cal A$ returns $(S,\bar S)$ and we correctly decide that $\phi^* > \gamma(2n)$. 
Now assume that $\phi^* < 1$. Denote the sparsest cut in $G$ by $A$. Let $T = \set{u_1:u\in A}\cup \set{u_2:u\notin A}$.
We have,
\ifSODA
\begin{align*}
\mathrlap{w(E(S,\bar S)) - w(E(T,\bar T))}\quad &
 \\
&= w(E(S\cap T, \bar S\cap T)) + w(E(S\cap \bar T, \bar S\cap \bar T)) \\
&\quad - 
(w(E(S\cap T, S\cap \bar T)) + w(E(\bar S\cap T, \bar S\cap \bar T)))\\
&=2\, \mathrm{cap}(E_c(A,\bar A)) -  2\, \mathrm{dem}(E_d(A,\bar A)) < 0.
\end{align*}
\else
\begin{align*}
w(E(S,\bar S)) - w(E(T,\bar T)) &= w(E(S\cap T, \bar S\cap T)) + w(E(S\cap \bar T, \bar S\cap \bar T)) \\
&\quad - 
(w(E(S\cap T, S\cap \bar T)) + w(E(\bar S\cap T, \bar S\cap \bar T)))\\
&=2\, \mathrm{cap}(E_c(A,\bar A)) -  2\, \mathrm{dem}(E_d(A,\bar A)) < 0.
\end{align*}
\fi
Hence $(S,\bar S)$ is not a maximum cut. Since $\cal A$ is a robust algorithm it must either return a cut different from $(S,\bar S)$
or $\perp$. Therefore, we decide that $\phi^* > \gamma(2n)$. 

\end{proof}

We get as a corollary that if there is a robust polynomial-time algorithm for $\gamma$-stable instances of Max Cut then there is a $\gamma(2n)$-approximation
algorithm for Sparsest Cut (the algorithm finds the value of Sparsest Cut but
not the actual cut). 
\begin{corollary}
Suppose that there is a robust polynomial-time algorithm $\cal A$ for $\gamma$-stable instances of Max Cut with $\gamma \geq \gamma(n)$.
Then there is a polynomial-time algorithm for the decision version of Sparsest Cut 
that given an instance with value $\phi^*$ and $\varepsilon > 0$
outputs a value $\phi_{\text{approx}}$ between $(1-\varepsilon)\phi^*/\gamma(2n)$ 
and $\phi^*$.
\end{corollary}
\begin{proof}
Let $\phi_{ARV}$ be the approximate value of the problem given by the algorithm of Arora, Rao and Vazirani. We try all possible values of $\phi_0$ of the form
$(1+k \varepsilon) \phi_{ARV}$ in the range $(\phi_{ARV},(\alpha_{ARV}+\varepsilon)\phi_{ARV})$. For each value, we
run the algorithm $\cal B$ from Theorem~\ref{thm:hardness}. We find the smallest value $\phi_{\text{approx}}'$ of $\phi_0$ such that $\cal B$ returns that $\phi^* < \phi_0$.
Note that if $\phi_0 > \phi^*$ then the promise of Theorem~\ref{thm:hardness} 
is satisfied and thus the algorithm $\cal B$ 
returns that $\phi^* < \phi_0$. Therefore, $\phi_{\text{approx}}' \leq (1+\varepsilon) \phi^*$. 

Similarly, if $\phi_0 < \phi_*/\gamma(2n)$ then the promise is satisfied and thus 
$\cal B$ returns that $\phi_* >  \gamma(2n) \phi_0$. Therefore, $\phi_{\text{approx}}'
\geq \phi^*/\gamma(2n)$. We output $\phi_{\text{approx}} = \phi_{\text{approx}}'/ (1+\varepsilon)$.
\end{proof}

We note that Theorem~\ref{thm:hardness} implies that there is no polynomially-time tractable relaxation
for Max Cut that is integral on $\gamma$-stable instances if
$\gamma < \alpha_{SC}(n/2)$. If there was such a relaxation, by solving it, 
we would get a robust algorithm as we do in Corollary~\ref{cor:algorithm}.
\begin{corollary}
There is no polynomial-time tractable relaxation for Max Cut that is integral on $\gamma$-stable instances if
$\gamma < \alpha_{SC}(n/2)$.
\end{corollary}

\subsection{SDP Integrality Gap} \label{sec:intgap}
In this section, we prove that the SDP relaxation for $\gamma$-stable instances is not integral in general when 
$\gamma <D_{\ell_2^2\to \ell_1}(n)$. To this end, we show how to transform an integrality gap Sparsest Cut instance 
to a $\gamma$-stable Max Cut instance with a non-integral SDP solution.

%\ifconf
%We will need the following technical lemma.
%We prove it in Appendix~\ref{sec:proof-special-SDP}.
%\else
We say that an instance of Sparsest Cut has integrality gap $D > 1$ if for some assignment of vectors $u \mapsto \bar u$ such that
the set $\{\bar u\}$ is an $\ell_2^2$ space, we have 
\ifSODA
\begin{align*}
\phi(A) &{}= \frac{\mathrm{cap}(E_c(A,\bar A))}{\mathrm{dem}(E_d(A,\bar A))} \\
&{}\geq D \cdot \frac{\sum_{(u,v)\in E_c} \mathrm{cap}_{uv}\cdot\|\bar u - \bar v\|^2}{\sum_{(u,v)\in E_d} \mathrm{dem}_{uv}\cdot\|\bar u - \bar v\|^2}
\end{align*}
for every cut $(A, \bar A)$.
\begin{fact}{\sc (Linial, London and Rabinovich~\cite{LLR}, Aumann and Rabani~\cite{AR})}
\else
$$
\phi(A) = \frac{\mathrm{cap}(E_c(A,\bar A))}{\mathrm{dem}(E_d(A,\bar A))} \geq D \cdot
\frac{\sum_{(u,v)\in E_c} \mathrm{cap}_{uv}\cdot\|\bar u - \bar v\|^2}{\sum_{(u,v)\in E_d} \mathrm{dem}_{uv}\cdot\|\bar u - \bar v\|^2}
\quad\text{for every cut} (A, \bar A).
$$
\begin{fact}[Linial, London and Rabinovich~\cite{LLR}, Aumann and Rabani~\cite{AR}]
\fi
\label{fact:integrality-gap-SC}
For every $n > 1$, there is an instance of Sparsest Cut on $n$ vertices with integrality gap $D_{\ell_2^2\to \ell_1}(n)$.
\end{fact}
We will need the following technical lemma. 
%\fi
%
\begin{lemma} \label{lem:special-SDP}
For every $n > 1$ and $\varepsilon > 0$, there exists a Sparsest Cut instance on $n$ vertices
and unit vectors $\{\bar u\}$ such that vectors $\{\pm \bar u\}$ form an $\ell_2^2$ space and 
\ifSODA for every cut $(A, \bar A)$\fi
\begin{align}
\mathrm{cap}(E_c(A,\bar A)) &> (D_{\ell_2^2\to \ell_1} - \varepsilon) \cdot \mathrm{dem}(E_d(A,\bar A)) 
\ifSODA\else\qquad\text{for every cut } (A, \bar A)\fi \label{ineq:combinatorial}\\
\sum_{(u,v)\in E_c} \mathrm{cap}_{uv}\cdot \|\bar u - \bar v\|^2 &< \sum_{(u,v)\in E_d} \mathrm{dem}_{uv}\cdot \|\bar u - \bar v\|^2 
\label{ineq:vector}.
\end{align}
\end{lemma}
To prove the lemma, we first rescale demands so that conditions~(\ref{ineq:combinatorial}) and~(\ref{ineq:vector}) hold. Then we
transform vectors $u$ so that all of them lie on the unit sphere. Specifically, if all vectors $u_i$ lie on some sphere,
we 
scale all vectors $u_i$ to unit vectors and move the origin to the center of the sphere; these transformations preserve ratios of
distances between vectors. In a degenerate case, when all vectors $u_i$ do not lie on a sphere, we first slightly perturb all vectors and
then apply the above argument. The formal proof is a bit technical, so we present it in Appendix~\ref{sec:proof-special-SDP}.
%\fi

\begin{theorem}\label{thm:sdp-gap}
For every $n$ and $\gamma\in [1,D_{\ell_2^2\to \ell_1}(n/2))$, there is a $\gamma$-stable instance $G$ of Max Cut
on $2n$ vertices, such that that the SDP relaxation for $G$ is not integral.
\end{theorem}
\begin{proof}
Let $\varepsilon = (D_{\ell_2^2\to \ell_1}(n) - \gamma)/2$.
From Lemma~\ref{lem:special-SDP}, we get a Sparsest Cut instance on a set $V_0$ and vectors $\bar u$ such that 
$\set{\pm \bar u}$ is an $\ell_2^2$ space
and inequalities (\ref{ineq:combinatorial}) and (\ref{ineq:vector}) hold. 
We apply reduction from Section~\ref{sec:reduction} to this instance and obtain a graph $G(V,E, w)$.
From Lemma~\ref{lem:reduction} and inequality~(\ref{ineq:combinatorial}), we get that $G$ is a $\gamma$-stable Max Cut instance.

We define an SDP solution for the SDP relaxation for $G$ by $\bar u_1 = \bar u$ and $\bar u_2 = -\bar u$. Since all vectors $\bar u$ are unit vectors and 
$\set{\pm \bar u}$ is an $\ell_2^2$ space, this is a feasible SDP solution.
Its value equals
\ifSODA
\begin{align*}
SDP &\equiv \frac{1}{4} \sum_{(x,y) \in E} w(x,y) \|\bar x- \bar y\|^2 = n W_{\infty} \\
       &\ {} + \frac{1}{4} \sum_{(u,v) \in E_d} \mathrm{dem}_{uv} \cdot(\|\bar u_1- \bar v_1\|^2 + \|\bar u_2- \bar v_2\|^2) \\
       &\ {}+ \frac{1}{4} \sum_{(u,v) \in E_c} \mathrm{cap}_{uv} \cdot (\|\bar u_1- \bar v_2\|^2 + \|\bar u_2- \bar v_1\|^2)\\
       &= n W_{\infty} + \frac{1}{2} \sum_{(u,v) \in E_d} \mathrm{dem}_{uv} \cdot \|\bar u - \bar v\|^2  \\
       &\ + \frac{1}{2} \sum_{(u,v) \in E_c} \mathrm{cap}_{uv} \|\bar u + \bar v\|^2.
\end{align*}
\else
\begin{align*}
SDP &\equiv \frac{1}{4} \sum_{(x,y) \in E} w(x,y) \|\bar x- \bar y\|^2 \\
       &= n W_{\infty} + \frac{1}{4} \sum_{(u,v) \in E_d} \mathrm{dem}_{uv} \cdot(\|\bar u_1- \bar v_1\|^2 + \|\bar u_2- \bar v_2\|^2) \\
       &\phantom{{}= n W_{\infty}} {}+ \frac{1}{4} \sum_{(u,v) \in E_c} \mathrm{cap}_{uv} \cdot (\|\bar u_1- \bar v_2\|^2 + \|\bar u_2- \bar v_1\|^2)\\
       &= n W_{\infty} + \frac{1}{2} \sum_{(u,v) \in E_d} \mathrm{dem}_{uv} \cdot \|\bar u - \bar v\|^2  + \frac{1}{2} \sum_{(u,v) \in E_c} \mathrm{cap}_{uv} \|\bar u + \bar v\|^2.
\end{align*}
\fi
Using that $\|\bar u + \bar v\|^2 = 4- \|\bar u- \bar v\|^2$, we get 
\ifSODA
$SDP = (n W_{\infty} + 2 \mathrm{cap} (E_c)) +  \frac{1}{2}\sum_{(u,v) \in E_d} \mathrm{dem}_{uv} \cdot \|\bar u - \bar v\|^2 -  \frac{1}{2}\sum_{(u,v) \in E_c} \mathrm{dem}_{uv} \|\bar u - \bar v\|^2$.
\else
$$SDP = (n W_{\infty} + 2 \mathrm{cap} (E_c)) +  \frac{1}{2}\sum_{(u,v) \in E_d} \mathrm{dem}_{uv} \cdot \|\bar u - \bar v\|^2 -  \frac{1}{2}\sum_{(u,v) \in E_c} \mathrm{dem}_{uv} \|\bar u - \bar v\|^2.$$
\fi
The first term $n W_{\infty} + 2 \mathrm{cap} (E_c)$ equals $w(S,\bar S)$ (where $(S,\bar S)$ is maximum cut).
From inequality~(\ref{ineq:vector}), we get 
\ifSODA
$SDP = w(S, \bar S) +  \frac{1}{2}\sum_{(u,v) \in E_d} \mathrm{dem}_{uv} \cdot \|\bar u - \bar v\|^2 -  \frac{1}{2}\sum_{(u,v) \in E_c} \mathrm{dem}_{uv} \|\bar u - \bar v\|^2 > w(S,\bar S)$.
\else
$$SDP = w(S, \bar S) +  \frac{1}{2}\sum_{(u,v) \in E_d} \mathrm{dem}_{uv} \cdot \|\bar u - \bar v\|^2 -  \frac{1}{2}\sum_{(u,v) \in E_c} \mathrm{dem}_{uv} \|\bar u - \bar v\|^2 > w(S,\bar S).$$
\fi
We conclude that the optimal SDP solution has value at least $SDP$, which is greater than $w(S,\bar S)$.
Therefore, the SDP relaxation is not integral.

\end{proof}

\subsection{Hardness Result for Max \texorpdfstring{$k$}{k}-Cut}
\label{sec:max-k-cut}
In this section, we prove a hardness result for Max $k$-Cut.  
\begin{Definition}
The Max $k$-Cut problem is to partition a given weighted graph $G$ into $k$ pieces so as to maximize the total weight of cut edges.
\end{Definition}

\begin{Definition}
Let us say that an instance $G=(V,E,w)$ of Max $k$-Cut is $\infty$-stable if it is $\gamma$-stable for every $\gamma$.
That is, there is a partition $\cal P$ of $V$ such that for every set of positive weights $w'$, $\cal P$ is an optimal solution
for Max Cut instance $G'=(V,E,w')$.
\end{Definition}

\begin{claim}\label{claim:neg-max-k-cut}
For every $k \geq 3$, there is no polynomial-time algorithm that solves $\infty$-stable instances of Max $k$-Cut unless $NP=RP$.
\end{claim}
\begin{proof}
The claim easily follows from the hardness result for the Unique $k$-Coloring problem by Barbanchon~\cite{Barbanchon}.
Recall that a graph $G$ is uniquely $k$ colorable if there exists exactly one proper coloring of $G$ in $k$ colors
(up to permutation of the colors).
Barbanchon~\cite{Barbanchon} showed%
\footnote{Barbanchon~\cite{Barbanchon} states his result only for $k=3$. The result for $k > 3$ follows
  from his result as follows. For a graph $G$, let $G'$ be the union of graphs $G$ and $K_{k-3}$ in which every vertex of $G$ is connected
  with every vertex of $K_{k-3}$. Then $G$ is uniquely 3-colorable if and only if $G'$ is uniquely $k$-colorable.}
that there is no polynomial algorithm that given a uniquely $k$-colorable graph finds its $k$ coloring unless $NP=RP$.

Let $G$ be a uniquely $k$-colorable graph. We assign each edge of $G$ weight $1$ and obtain an instance of Max $k$-Cut.
We show that the instance is $\infty$-stable. Let $\cal P$ be the partition corresponding to 
the unique $k$-coloring $\cal C$ of $G$. Note that no matter what positive weights we assign to edges,
the value of $\cal P$ equals the total weight of all edges in the graph (since $\cal P$ cuts all edges).
Thus $\cal P$ is an optimal $k$-partition. Moreover, $\cal P$ is the only optimal partition. Indeed
if a $k$-partition ${\cal P}'$ cuts all edges, then the coloring that colors every piece in ${\cal P}'$
in its own color is a proper $k$-coloring, and thus it is equal to $\cal C$ (up to permutation of the colors).
The result of Barbanchon implies that there is no polynomial-time algorithm that finds the optimal 
Max $k$-Cut in $G$ unless $NP=RP$.
\end{proof}

%%algorithm for multiway cut

\section{Weakly Stable Instances}\label{sec:weakstability}
\subsection{Weakly Stable Instances of Max Cut}
In this section, we define a relaxed notion of stability, which we call weak stability, and give an algorithm for approximately solving weakly stable
instances of Max Cut.
We note that 
Awasthi, Blum and Sheffet~\cite{Awasthietal} and Balcan and Liang~\cite{BalcanL} studied a very closely related notion of \textit{perturbation resilience} for the $k$-Median Clustering problem.
\begin{Definition}\label{def:weak-stabilityA}
Consider a weighted graph $G=(V,E,w)$. Let $(S,\bar S)$ be a maximum cut in $G$, $N$ be a set of cuts that contains $(S,\bar S)$, and $\gamma \geq 1$.
We say that $G$ is a  $(\gamma, N)$-weakly stable instance of Max Cut 
if for every $\gamma$-perturbation $G'=(V,E,w')$ of $G$, we have
$$w'(E(S,\bar S)) \geq w'(E(T,\bar T)).$$
\end{Definition}
This definition is equivalent to the following definition (see Appendix~\ref{apx:equivalence} for the proof).
\begin{Definition}\label{def:weak-stability}
Consider a weighted graph $G=(V,E,w)$. Let $(S,\bar S)$ be a maximum cut in $G$, $N$ be a set of cuts that contains $(S,\bar S)$, and $\gamma \geq 1$.
We say that $G$ is a  $(\gamma, N)$-weakly stable instance of Max Cut if for every cut $(T, \bar T)\notin N$:
$$w(E(S, \bar S)\setminus E(T, \bar T)) > \gamma\cdot w(E(T, \bar T )\setminus E(S, \bar S)).$$
\end{Definition}
The notion of weak stability generalizes the notion of stability: an instance is $\gamma$-stable if and only if it is $(\gamma, \set{(S,\bar S)})$-weakly
stable. We think of the set $N$ in the definition of weak stability as a neighborhood of the maximum cut $(S, \bar S)$; it contains cuts that are ``close enough'' to $(S, \bar S)$. Intuitively, the definition requires that every cut that is sufficiently different from $(S,\bar S)$ is much smaller then $(S,\bar S)$, but does not impose any restrictions on cuts that are close to $(S,\bar S)$. One natural way to define the neighborhood of $(S,\bar S)$
is captured in the following definition.
\begin{Definition}\label{def:weak-stability-delta}
Consider a weighted graph $G$. Let $(S,\bar S)$ be a maximum cut in $G$, $\delta \geq 0$, and $\gamma \geq 1$.
We say that $G$ is a  $(\gamma, \delta)$-weakly stable instance of Max Cut if $G$ is $(\gamma, \set{(S',\bar S'): |S\Delta S'| \leq \delta n})$-weakly
stable.
In other words,  $G$ is  $(\gamma, \delta)$-weakly stable if 
for every cut $(T, \bar T)$ such that $|S \Delta T| > \delta n$ and $|S \Delta {\bar T}| > \delta n$, we have
$$w(E(S, \bar S)\setminus E(T, \bar T)) > \gamma\cdot w(E(T, \bar T )\setminus E(S, \bar S)).$$
\end{Definition}

%A notion similar to weak stability was also studied by Balcan and Liang \cite{BalcanL} 
%(We note that a $(\gamma,\delta)$-approximation--22stable instance of a combinatorial optimization problem is $(\gamma,\delta)$ weakly stable.)
The main result of this section is the following theorem.
\begin{theorem}\label{thm:weakly-stable}
There is a polynomial-time algorithm that given a $(\gamma, N)$-stable instance of Max Cut, returns a cut from $N$ if $\gamma \geq c\sqrt{\log n}\log \log n$ (for some absolute constant $c$).
(The set $N$ is not part of the input and is not known to the algorithm.)
\end{theorem}

The algorithm starts with an arbitrary cut and then iteratively improves it. We now describe a subroutine that algorithm runs in each iteration.

\begin{lemma}\label{lem:cut-improvement}
There is a polynomial-time algorithm $\cal A$ for the following task.
Let $G=(V,E,w)$ be a $(\gamma, \N)$-weakly stable instance of Max Cut for some $N$ and $\gamma \geq c\sqrt{\log n}\log\log n$ (where $c$ is an absolute constant), $(S,\bar S)$ be the optimal cut, and $(T,\bar T)\notin \N$ be a cut in $G$. 
The algorithm $\cal A$ given the graph $G$, the cut $(T,\bar T)$ and a parameter $\omega \in {\mathbb R}^+$  
either returns a cut $(T', \bar T')$ such that $w(T',\bar T') \geq w(T,\bar T) + \omega$, or returns $\perp$.
The algorithm always returns a cut $(T', \bar T')$ if $w(E(S,\bar S) \setminus E(T,\bar T)) \geq 4m\omega$ (where $m= |E|$). 
\end{lemma}
\begin{proof}
We construct an auxiliary Sparsest Cut instance $\cal I$ on $V$ defined by
\ifSODA
\begin{itemize}
\item $E_c = E(T,\bar T)$, \item $ \mathrm{cap}_{uv} = w(u,v)$,
\item $E_d = \set{(u,v) \in E\setminus E(T,\bar T):  w(u,v) \geq 2\omega}$ , \item $ \mathrm{dem}_{uv} = w(u,v)$.
\end{itemize}
\else
\begin{align*}
E_c &= E(T,\bar T) , & \mathrm{cap}_{uv} &= w(u,v)\\
E_d &= \set{(u,v) \in E\setminus E(T,\bar T):  w(u,v) \geq 2\omega} , & \mathrm{dem}_{uv} &= w(u,v).
\end{align*}
\fi
Then we run the approximation algorithm for Sparsest Cut by Arora, Lee and Naor and find an approximate cut $(A, \bar A)$.
We let $T' = (T \cap A) \cup (\bar T \cap \bar A)$ (see~Figure~\ref{fig:improvement}). If $w(T',\bar T') \geq w(T,\bar T) + \omega$ then we return $(T', \bar T')$,
otherwise we return $\perp$. 

\begin{figure}
\centering
%\begin{center}
\begin{tikzpicture}
    \draw[very thick] (-0.9,0.25) to (-3.1, 0.25) to [out=90,in=180] (-2, 2) to [out=0,in=90] (-0.9, 0.25);
    \draw[very thick] (0.9,-0.25) to (3.1, -0.25) to [out=-90,in=0] (2, -2) to [out=180,in=-90] (0.9, -0.25);

    \draw[very thick, dashed, color=black!40!blue] (+0.9,0.25) to (+3.1, 0.25) to [out=90,in=0] (+2, 2) to [out=180,in=90] (+0.9, 0.25);
    \draw[very thick, dashed, color=black!40!blue] (-0.9,-0.25) to (-3.1, -0.25) to [out=-90,in=180] (-2, -2) to [out=0,in=-90] (-0.9, -0.25);

    \foreach \x in {-1.7,-1.9,-2.1, -2.3, 1.7, 1.9, 2.1, 2.3}
       \draw[thick, color=black!40!green] (\x, 0.45) -- (\x,-0.45);
    
    \draw[thick, color=black] (-1.2+0.0875, 0.85) -- (1.2+0.0875, -0.55);
    \draw[thick, color=black] (-1.2, 0.7) -- (1.2, -0.7);
    \draw[thick, color=black] (-1.2-0.0875, 0.55) -- (1.2-0.0875, -0.85);

    \draw[thick, color=black] (-1.2+0.0875, -0.85) -- (1.2+0.0875, 0.55);
    \draw[thick, color=black] (-1.2, -0.7) -- (1.2, 0.7);
    \draw[thick, color=black] (-1.2-0.0875, -0.55) -- (1.2-0.0875, 0.85);

    \node[above] at (-2,0.6) {\Large $T\cap A$};
    \node[above] at (2,0.6) {\Large $\bar T\cap A$};
    \node[below] at (-2,-0.6) {\Large $T\cap  \bar A$};
    \node[below] at (2,-0.6) {\Large $\bar{T}\cap \bar A$};

\end{tikzpicture}
%\includegraphics[scale=0.75]{figure-2-crop.pdf}
%\end{center}
\caption{This figure shows the Sparsest Cut instance $\cal I$ and cut $(A,\bar A)$.
Demand pairs separated by $(A,\bar A)$ are shown by vertical segments,
capacity edges cut by $(A,\bar A)$ are shown by diagonal segments.
Set $T'$ consists of vertices in semicircles in the upper-left and lower-right corners. }
\label{fig:improvement}
\end{figure}
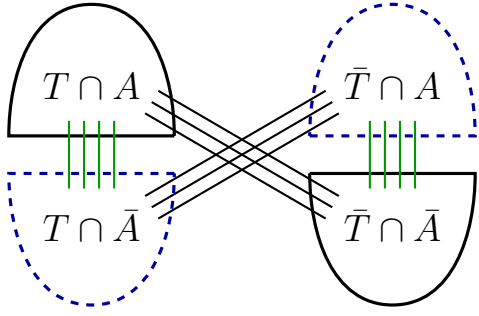

Whenever the algorithm returns a cut, the cut satisfies the requirement $w(T',\bar T') \geq w(T,\bar T) + \omega$. Thus 
we only need to prove that if $w(E(S,\bar S) \setminus E(T,\bar T)) \geq 4m\omega$ then the algorithm finds a cut.

First we show that there is a Sparsest Cut with sparsity at most $2/\gamma$ in $\cal I$. Let $A^* = (S \cap T) \cup (\bar S \cap \bar T)$.
Since $(T,\bar T) \notin \N$, we have $w(E(S, \bar S)\setminus E(T, \bar T)) > \gamma\cdot w(E(T, \bar T )\setminus E(S, \bar S))$, or equivalently
$w((E\setminus E(T,\bar T)) \cap E(A^*,\bar A^*)) > \gamma \cdot w(E(T,\bar T) \cap E(A^*,\bar A^*))$. Thus
\ifSODA
\begin{align*}
\mathrm{dem}&(E_d \cap E(A^*,\bar A^*)) \geq w((E\setminus E(T,\bar T)) \cap E(A^*,\bar A^*)) \\
&\quad{}- m \cdot 2\omega \geq w((E\setminus E(T,\bar T))/2  \\
&> \gamma \cdot  w(E(T,\bar T) \cap E(A^*,\bar A^*))/2 = (\gamma/2) \cdot \mathrm{cap}(E_c).
\end{align*}
\else
\begin{align*}
\mathrm{dem}(E_d \cap E(A^*,\bar A^*)) &\geq w((E\setminus E(T,\bar T)) \cap E(A^*,\bar A^*)) - m \cdot 2\omega \geq w((E\setminus E(T,\bar T))/2  \\
&> \gamma \cdot  w(E(T,\bar T) \cap E(A^*,\bar A^*))/2 = (\gamma/2) \cdot \mathrm{cap}(E_c).
\end{align*}
\fi
We get that $\phi(A^*) < 2/\gamma$. Therefore, our algorithm finds a cut $A$ with $\phi(A) \leq O(\sqrt{\log n}\log\log n) \cdot  2/\gamma < 1/2$.
We have
\ifSODA
\begin{align*}
w&(T',\bar T') - w(T,\bar T) \\
&= w(E(T',\bar T')\setminus E(T,\bar T)) - w(E(T,\bar T)\setminus E(T',\bar T')) \\
&\geq \mathrm{dem}(E_d \cap E(A,\bar A)) - w(E_c \cap E(A,\bar A)) \\
&\geq  \mathrm{dem}(E_d \cap E(A,\bar A))/2 \geq
\omega.
\end{align*}
\else
\begin{align*}
w(T',\bar T') - w(T,\bar T) &= w(E(T',\bar T')\setminus E(T,\bar T)) - w(E(T,\bar T)\setminus E(T',\bar T')) \\
&\geq \mathrm{dem}(E_d \cap E(A,\bar A)) - w(E_c \cap E(A,\bar A)) \geq  \mathrm{dem}(E_d \cap E(A,\bar A))/2 \geq
\omega.
\end{align*}
\fi
We used that the demand of every pair in $E_d$ is at least $2\omega$, and hence $\mathrm{dem}(E_d \cap E(A,\bar A))/2 \geq
\omega$.
\end{proof}
Now we are ready to prove Theorem~\ref{thm:weakly-stable}.
\ifSODA\\[0.5em] \noindent \textit{Proof of Theorem~\ref{thm:weakly-stable}.}\\[0.25em]
\else
\begin{proof}[Proof of Theorem~\ref{thm:weakly-stable}] {\ }\\
\fi
\indent\textbf{Algorithm.} 
We start with an arbitrary cut $(T,\bar T)$. Then we iteratively run the algorithm $\cal A$ 
from Lemma~\ref{lem:cut-improvement}. 
In each iteration, we go over all values of $\omega$ in $\Omega = \{w(u,v)/(4m) : (u,v) \in E\}$ in the descending order,
and execute $\cal A$ on input $G$, $(T,\bar T)$ and $\omega$. If $\cal A$ 
finds a cut $(T', \bar T')$, we let $T = T'$ and start a new iteration. If $\cal A$ does not find any cut, we stop and output $(T,\bar T)$.

\textbf{Analysis.} 
We first show that the algorithm always returns a cut from $N$. At every step of the algorithm when $(T, \bar T)\notin \N$, we have $w(S,\bar S) > w(T, \bar T)$, hence
$E(S,\bar S)\setminus E(T, \bar T)\neq \varnothing$ and $w(E(S,\bar S)\setminus E(T, \bar T)) \geq \min_{e\in E} w(e)$. Therefore,
for some $\omega\in \Omega$ (in particular, for $\omega = \frac{\min_{e\in E} w(e)}{4m}$; see the statement of Lemma~\ref{lem:cut-improvement}), the algorithm $\cal A$ finds a
better cut $(T',\bar T')$, and the main algorithm does not terminate. It remains to check that the running time is polynomial.
%%It remains to prove that the algorithm terminates after polynomially many steps.

Consider one iteration of the algorithm. Let $(T,\bar T)$ be the current cut. Let $(u,v)$ be the heaviest edge in $E(S,\bar S) \setminus E(T,\bar T)$
and $\omega^* = w(u,v) / (4m)$. Note that in this iteration we find a cut when we run $\cal A$ with some $\omega \geq \omega^*$
(since if we do not find a cut $(T',\bar T')$ when $\omega > \omega^*$, we must find a cut $(T', \bar T')$ when $\omega = \omega^*$ by
Lemma~\ref{lem:cut-improvement}). We also have
\ifSODA
$w(E(S,\bar S)) - w(E(T,\bar T)) \leq w(E(S,\bar S)\setminus E(T,\bar T)) \leq  |E(S,\bar S)\setminus E(T,\bar T)| 
\cdot (4m \omega^*) \leq 4m^2 \omega$.
\else
$$w(E(S,\bar S)) - w(E(T,\bar T)) \leq w(E(S,\bar S)\setminus E(T,\bar T)) \leq  |E(S,\bar S)\setminus E(T,\bar T)| 
\cdot (4m \omega^*) \leq 4m^2 \omega.$$
\fi
We charge this iteration to ``level'' $\omega$. We show that every $\omega \in \Omega$ pays for at most $4m^2$ iterations and therefore the number of iterations is $O(m^3)$.

Indeed, consider $\omega \in \Omega$. Let $T_0$ be the value of $T$ just before we perform 
an iteration at level $\omega$ for the first time, and $T_k$ be the value of $T$ right after we perform $k$ iterations at level $\omega$ 
(possibly we perform iterations at other levels in between).
We have,
\ifSODA
$w(E(T_k,\bar T_k)) \geq w(E(T_0,\bar T_0)) +k\omega \geq w(E(S,\bar S)) -  4m^2\omega + k\omega$.
\else
$$w(E(T_k,\bar T_k)) \geq w(E(T_0,\bar T_0)) +k\omega \geq w(E(S,\bar S)) -  4m^2\omega + k\omega.$$
\fi
Since $w(E(T_k,\bar T_k)) \leq w(E(S,\bar S))$, we get that $k \leq 4m^2$.
This concludes the proof.
\ifSODA
{\hfill\vbox{\hrule height.2pt\hbox{\vrule width.2pt height5pt \kern5pt \vrule width.2pt} \hrule height.2pt}}
\else
\end{proof}
\fi

\subsection{Weakly Stable Instances of Minimum Multiway Cut}
In this section, we give an algorithm for approximately solving weakly stable instances of Minimum Multiway Cut.

\begin{Definition}\label{def:weak-stability-multiwayA}
Consider a weighted graph $G$. Let ${\cal S}^* = (S_1^*,\dots, S_k^*)$ be a minimum multiway cut in $G$, $N$ be a set of multiway 
cuts that contains ${\cal S}^*$, and $\gamma \geq 1$.
We say that $G$ is a  $(\gamma, N)$-weakly stable instance of Minimum Multiway Cut if for 
$\gamma$-perturbation $G'=(V,E,w')$ of $G$ and every multiway cut 
${\cal S}' = (S_1',\dots, S_k')\notin N$:
$$w'(E') > w'(E^*),$$
where $E^*$ is the set of edges cut by ${\cal S}^*$ and $E'$ is the set of edges cut by ${\cal S}'$.
\end{Definition}

This definition is equivalent to the following definition (see Appendix~\ref{apx:equivalence} for the proof).

\begin{Definition}\label{def:weak-stability-multiway}
Consider a weighted graph $G$. Let ${\cal S}^* = (S_1^*,\dots, S_k^*)$ be a minimum multiway cut in $G$, $N$ be a set of multiway 
cuts that contains ${\cal S}^*$, and $\gamma \geq 1$.
We say that $G$ is a  $(\gamma, N)$-weakly stable instance of Minimum Multiway Cut if for every multiway cut 
${\cal S}' = (S_1',\dots, S_k')\notin N$:
$$w(E'\setminus E^*) > \gamma\cdot w(E^*\setminus E'),$$
where $E^*$ is the set of edges cut by ${\cal S}^*$ and $E'$ is the set of edges cut by ${\cal S}'$.
\end{Definition}

\noindent The main result of this section is the following theorem.
\begin{theorem}\label{thm:weakly-stable-multiway-prime}
There is a polynomial-time algorithm for the following task.
Given a $(4, N)$-stable instance of Minimum Multiway Cut with integer edge weights in the range $[1, poly(n)]$,
the algorithm returns a multiway cut from $N$.
The set $N$ is not part of the input and is not known to the algorithm.
\end{theorem}
If the weights are not polynomially bounded, the following version of this theorem holds. (The proofs of Theorems~\ref{thm:weakly-stable-multiway-prime} and \ref{thm:weakly-stable-multiway} are very similar. For simplicity of exposition, we only present the proof of Theorem~\ref{thm:weakly-stable-multiway-prime}.)
\begin{theorem}\label{thm:weakly-stable-multiway}
There is a polynomial-time algorithm that given a $(4 + \varepsilon, N)$-stable instance of Minimum Multiway Cut, returns a multiway
cut from $N$. The running time of the algorithm is inversely proportional do $\varepsilon$.
The set $N$ is not part of the input and is not known to the algorithm.
\end{theorem}

The algorithm starts with an arbitrary multiway cut $\cal S$ and then iteratively improves it.

\begin{lemma}\label{lem:improve-multiway}
There is a polynomial-time algorithm that given a $(4,N)$ weakly stable instance $G=(V,E,w)$ and 
a solution ${\cal S}^{\circ}$ either finds a solution ${\cal S}'$ of smaller cost or 
certifies that ${\cal S}^{\circ} \in N$.
\end{lemma}
\begin{proof}
Let $E^{\circ}$ be the set of edges cut by ${\cal S}^{\circ}$.
Define edge weights $w'(u,v)$ by 
$$
w'(u,v) = 
\begin{cases}
w(u,v), & \text{if} (u,v) \in E^{\circ},\\
4\,w(u,v), & \text{otherwise}.
\end{cases}
$$
We solve the LP relaxation~(\ref{eq:CKR-relaxation}) for Multiway Cut with weights $w'(u,v)$.
Let $\{\bar u\}$ be the LP solution. 
Consider the distribution of random cuts ${\cal S}' = (S_1',\dots, S_k')$
from Lemma~\ref{lem:hol-rounding}.
Let $E'$ be the set of edges cut by ${\cal S}'$. Similarly to the proof of Theorem~\ref{thm:main-multiway}, we define
\ifSODA
\begin{align*}
\mathsf{LP}_{+} &= \sum_{(u,v)\in E^{\circ}} w'(u,v) (1 - d(u,v)) \\&= \sum_{(u,v)\in E^{\circ}} w(u,v) (1 - d(u,v)),\\
\intertext{and}
\mathsf{LP}_{-} &= \sum_{(u,v)\in E\setminus E^{\circ}} w'(u,v) \, d(u,v) \\&= 4\sum_{(u,v)\in E\setminus E^{\circ}} w(u,v) \, d(u,v).
\end{align*}
\else
\begin{align*}
\mathsf{LP}_{+} &= \sum_{(u,v)\in E^{\circ}} w'(u,v) (1 - d(u,v)) = \sum_{(u,v)\in E^{\circ}} w(u,v) (1 - d(u,v)),\\
\mathsf{LP}_{-} &= \sum_{(u,v)\in E\setminus E^{\circ}} w'(u,v) \, d(u,v) = 4\sum_{(u,v)\in E\setminus E^{\circ}} w(u,v) \, d(u,v).
\end{align*}
\fi
We have, 
%$\mathsf{LP}_{+} - \mathsf{LP}_{-} \geq 0$ since the LP value is at most the value of the solution ${\cal S}^{\circ}$, and
\ifSODA
\begin{align*}
\E{w(E^{\circ}\setminus E')} &= \sum_{(u,v) \in E^{\circ}} w(u,v) \Pr((u,v) \text{ is not cut}) \\
&
\geq 
\sum_{(u,v) \in E^{\circ}} w(u,v) (1-d(u,v))/2 \\
&= \mathsf{LP}_{+}/2,\\
\intertext{and}
\E{w(E'\setminus E^{\circ})} &= \sum_{(u,v) \in E\setminus E^*} w(u,v) \, \Pr((u, v) \text{ is cut}) \\
&\leq 
2\sum_{(u,v) \in E^{\circ}} w(u,v) d(u,v) = \mathsf{LP}_{-}/2.
\end{align*}
\else
\begin{align*}
\E{w(E^{\circ}\setminus E')} &= \sum_{(u,v) \in E^{\circ}} w(u,v) \Pr((u,v) \text{ is not cut}) %\\
%&
\geq 
\sum_{(u,v) \in E^{\circ}} w(u,v) (1-d(u,v))/2 = \mathsf{LP}_{+}/2,\\
\E{w(E'\setminus E^{\circ})} &= \sum_{(u,v) \in E\setminus E^*} w(u,v) \, \Pr((u, v) \text{ is cut}) \leq 
2\sum_{(u,v) \in E^{\circ}} w(u,v) d(u,v) = \mathsf{LP}_{-}/2.
\end{align*}
\fi
Therefore, 
\ifSODA
\begin{align*}
\E{w(E^{\circ})- w(E')} &=  \E{w(E^{\circ}\setminus E')} - \E{w(E'\setminus E^{\circ})}\\
&= \frac{\mathsf{LP}_{+}- \mathsf{LP}_{-}}{2}.
\end{align*}
\else
$$\E{w(E^{\circ})- w(E')} =  \E{w(E^{\circ}\setminus E')} - \E{w(E'\setminus E^{\circ})}= \frac{\mathsf{LP}_{+}- \mathsf{LP}_{-}}{2}.$$
\fi
We have,
\ifSODA
\begin{align*}
\mathsf{LP}_{+} &{}- \mathsf{LP}_{-} = w'(E^{\circ}) - \sum_{(u,v) \in E} w'(u,v) \, d(u,v) \\
&\geq w'(E^{\circ}) - w'(E^*)  = w'(E^{\circ} \setminus E^*) - w'(E^* \setminus E^{\circ}),
\end{align*}
\else
$$\mathsf{LP}_{+} - \mathsf{LP}_{-} = w'(E^{\circ}) - \sum_{(u,v) \in E} w'(u,v) \, d(u,v) \geq w'(E^{\circ}) - w'(E^*) 
= w'(E^{\circ} \setminus E^*) - w'(E^* \setminus E^{\circ}),$$
\fi
since the LP value of solution $\{\bar u\}$ is at most the value of solution $E^*$ (of the multiway instance with weights $w'$).
Now if ${\cal S}^{\circ}\notin N$ then
\ifSODA
\begin{align*}
\mathsf{LP}_{+} - \mathsf{LP}_{-} &= w'(E^{\circ} \setminus E^*) - w'(E^* \setminus E^{\circ})\\
& = w(E^{\circ} \setminus E^*) - 4 w(E^* \setminus E^{\circ}) > 0.
\end{align*}
\else
$$\mathsf{LP}_{+} - \mathsf{LP}_{-} = w'(E^{\circ} \setminus E^*) - w'(E^* \setminus E^{\circ})
 = w(E^{\circ} \setminus E^*) - 4 w(E^* \setminus E^{\circ}) > 0.
$$
\fi
Therefore, $\E{w(E')-w(E^{\circ})} < 0$ and for some multiway cut in the distribution, we have 
$w(E')<w(E^{\circ})$. Note that we can efficiently go over all multiway cuts in the support of the
distribution. If we find a multiway cut with $w(E') < w(E^{\circ})$, we return it; otherwise, 
we output that ${\cal S}^{\circ}\in N$.
\end{proof}

\ifSODA\vspace{0.2cm}\noindent \textit{Proof of Theorem~\ref{thm:weakly-stable-multiway-prime}.}\\[0.25em]
\else
\begin{proof}[Proof of Theorem~\ref{thm:weakly-stable-multiway-prime}]
\fi
We start with an arbitrary feasible multiway cut ${\cal S}^{\circ}$ and iteratively 
improve it using the algorithm from Lemma~\ref{lem:improve-multiway}. Once the algorithm returns
that the current cut ${\cal S}^{\circ}$ lies in $N$, we output it.
Since the cost of the multiway cut decreases by at least $1$ in each iteration, and
the initial cost of ${\cal S}^{\circ}$ is polynomial in $n$, the algorithm terminates after
polynomially many steps.
\ifSODA

{\hfill\vbox{\hrule height.2pt\hbox{\vrule width.2pt height5pt \kern5pt \vrule width.2pt} \hrule height.2pt}}
\else
\end{proof}
\fi

\section{Discussion}\label{sec:Discussion}
In this paper, we presented algorithms for stable instances of Max Cut and Minimum Multiway Cut.
In conclusion, we briefly discuss what properties of these problems we used. 
We provide a sufficient condition under which there is an algorithm for stable instances
of a graph partitioning problem.

Consider a graph partitioning problem. Our goal is to partition a graph into several pieces, subject to certain constraints,
so as to minimize or maximize the weight of cut edges. Consider a metric relaxation for this problem. The relaxation
defines a metric $d(\cdot, \cdot)$ on the set of vertices. A combinatorial solution to the problem corresponds to a 
multicut metric $d(u,v)$: the distance between vertices in one piece is $0$, the distance between vertices in different
pieces is $1$. For Max Cut and Multiway Cut, 
we proved that \textit{the metric relaxation is integral when the instance is sufficiently stable};
this, in turn, implied the existence of polynomial-time robust algorithms for stable instances of these problems.
%\begin{property}[$\star$]
%The metric relaxation is integral when the instance is sufficiently stable. 
%\end{property}
We summarize the properties that we used in the proof in the following meta-theorem.
\begin{theorem}[Meta-theorem]
Consider a graph partitioning problem and a metric relaxation for it. Suppose that there is a rounding scheme 
that given a graph $G=(V,E,w)$ and a metric $d(\cdot, \cdot)$ returns a feasible partition such that
for some $\alpha \geq 1$ and $\beta \geq 1$:
\begin{tabbing}
\ifSODA
\=\  \=\quad \ \=\ \=\ \=\ \=\kill
\else
\quad\=\quad\=\quad\quad\=\quad\=\quad\=\quad\=\kill
\fi
\>\textnormal{\textbf{For a cut minimization problem,}} \\
\>\>$1.$ \> $\Prob{u\textnormal{ and }v\textnormal{ are separated}} \leq \alpha d(u,v)$, \\
\>\>$2.$ \> $\Prob{u\textnormal{ and }v\textnormal{ are not separated}} \geq \beta^{-1}(1-d(u,v))$.\\[0.3em]
\>\textnormal{\textbf{For a cut maximization problem,}} \\
\>\>$1'$.  \> $\Prob{u\textnormal{ and }v\textnormal{ are separated}} \geq \alpha^{-1} d(u,v)$  \\
\>\>$2'$.  \> $\Prob{u\textnormal{ and }v\textnormal{ are not separated}} \leq \beta (1-d(u,v))$
\end{tabbing}
Then the metric relaxation is integral for $(\alpha\beta)$-stable instances of the problem.
Consequently, there is a robust polynomial-time algorithm for $(\alpha\beta)$-stable instances
(if the relaxation is polynomial-time solvable). Moreover, there is an algorithm
for $(\alpha\beta + \varepsilon)$-weakly stable instances of the problem.
(The meta-theorem also holds for a cut maximization/minimization problem with positive and negative weights.
Then we require that all four properties $1$, $1'$, $2$ and $2'$ hold.)
\end{theorem}
The proof of this meta-theorem repeats the proofs of Theorems~\ref{thm:sdp-integral} and~\ref{thm:main-multiway}.
Note that if a rounding scheme just satisfies property 1 or $1'$ %(and not necessarily $2$ and $2'$)
then there is an $\alpha$ approximation algorithm for the problem. However, properties $1$ and $1'$
alone do not imply that the relaxation is integral. For example, there is a rounding scheme for Max $k$-Cut  
satisfying $1'$, 
but there is no algorithm for stable instances of Max $k$-Cut (see Claim~\ref{claim:neg-max-k-cut}). 
Another example is Minimum Multicut. There is a rounding scheme for the standard LP relaxation of 
Minimum Multicut with $\alpha = O(\log n)$~\cite{LR}. However, this relaxation is not integral even for 
$(n - 2 -\varepsilon)$-stable instances of the problem (for every $\varepsilon> 0$). Indeed, consider an instance on $(n-1)$ terminals $s_1,\dots,
s_{n-1}$ and one extra vertex $u$; $u$ is connected with $s_1$ by an edge of weight $n-2-\varepsilon/2$ and with all other terminals by
edges of weight $1$. We need to separate every pair of terminals $s_i$ and $s_j$.
This instance is $(n-2-\varepsilon)$-stable. However, the optimal LP solution
is not integral: it assigns $d(u,s_i) = 1/2$ (for every $i$) and $d(s_i,s_j) = 1$ (for every $i\neq j)$.
\bibliographystyle{plain}
\bibliography{stable}

\appendix

\section{Proof of Lemma~\ref{lem:special-SDP}}\label{sec:proof-special-SDP}
\ifSODA
\noindent \textit{Proof of Lemma~\ref{lem:special-SDP}.}
\else
\begin{proof}[Proof of Lemma~\ref{lem:special-SDP}.]
\fi
Consider an integrality gap instance with gap $D_{\ell_2^2\to\ell_1}(n)$(see Fact~$\ref{fact:integrality-gap-SC}$).
Denote the set of vertices by $V_0$, the set of capacity edges by $E_c$, and edge capacities by $\mathrm{cap}_{uv}$.
Similarly, denote the set of demand pairs by $E_d$, and demands by $\mathrm{dem}_{uv}$. Let $\{\bar u\}$ be the optimal SDP solution 
for the standard SDP relaxation for $\cal I$. Since the integrality gap is $D_{\ell_2^2\to \ell_1}(n)$, we have
that for every cut $(A, \bar A=V_0\setminus A)$
$$
\frac{\mathrm{cap}(E_c(A,\bar A))}{\mathrm{dem}(E_d(A,\bar A))} \geq D_{\ell_2^2\to \ell_1}\cdot
\frac{\sum_{(u,v)\in E_c} \mathrm{cap}_{uv}\cdot\|\bar u - \bar v\|^2}{\sum_{(u,v)\in E_d} \mathrm{dem}_{uv}\cdot\|\bar u - \bar v\|^2}.
$$
By rescaling demands we may assume without loss of generality that 
$$\frac{D_{\ell_2^2\to \ell_1}-\varepsilon}{D_{\ell_2^2\to \ell_1}} <
\frac{\sum_{(u,v)\in E_c} \mathrm{cap}_{uv}\cdot \|\bar u - \bar v\|^2}{\sum_{(u,v)\in E_d} \mathrm{dem}_{uv}\cdot \|\bar u - \bar v\|^2} < 1.$$
%$$\sum_{(u,v)\in E_d} \mathrm{dem}_{uv}\cdot \|\bar u - \bar v\|^2 = \sum_{(u,v)\in E_c} \mathrm{cap}_{uv}\cdot \|\bar u - \bar v\|^2.$$
Then we have (\ref{ineq:combinatorial}).

We now transform vectors $\bar u$ to unit vectors $\tilde u$ such that $\set{\pm \tilde u}$ is an $\ell_2^2$ space and 
vectors $\tilde u$ satisfy inequality~(\ref{ineq:vector}).
Choose a sufficiently small $\delta>0$ (which we specify later) and let $\hat u' = \hat u + \delta \cdot \bar e_u$, where vectors $\bar e_u$ are unit vectors
orthogonal to each other and all vectors $\bar v$. All vectors $\bar u'$ are in general position (that is, no $r+2$ vectors 
lie in an $r$ dimensional affine subspace). Therefore all vectors $\bar u'$ lie on some sphere. Denote its center by $\bar c$ and radius by
$R$. Let $\bar z$ be a unit vector orthogonal to all vectors $\bar u'$. Finally, define vectors $\tilde u$,
$$\tilde u = \frac{\sqrt{3}\,\bar z}{2}  + \frac{\bar u' - \bar c}{2R}.$$
Note that vectors $(\bar u' - \bar c)/R$ are unit vectors, and therefore vectors $\tilde u$ are also unit vectors.
Now 
$$\|\tilde u - \tilde v\|^2 = \left\|\frac{\bar u' - \bar v'}{2R}\right\|^2 = \frac{\|\bar u - \bar v'\|^2 + 2\delta^2}{4R^2}.$$
Therefore, when $\delta \to 0$,
\ifSODA
$\frac{\sum_{(u,v)\in E_c} \mathrm{cap}_{uv}\cdot \|\tilde u - \tilde v\|^2}{\sum_{(u,v)\in E_d} \mathrm{dem}_{uv}\cdot \|\tilde u - \tilde v\|^2}=
\frac{\sum_{(u,v)\in E_c} \mathrm{cap}_{uv}\cdot\|\bar u - \bar v\|^2}{\sum_{(u,v)\in E_d} \mathrm{dem}_{uv}\cdot\|\bar u - \bar v\|^2} + O(\delta^2)
$.
\else
$$\frac{\sum_{(u,v)\in E_c} \mathrm{cap}_{uv}\cdot \|\tilde u - \tilde v\|^2}{\sum_{(u,v)\in E_d} \mathrm{dem}_{uv}\cdot \|\tilde u - \tilde v\|^2}=
\frac{\sum_{(u,v)\in E_c} \mathrm{cap}_{uv}\cdot\|\bar u - \bar v\|^2}{\sum_{(u,v)\in E_d} \mathrm{dem}_{uv}\cdot\|\bar u - \bar v\|^2} + O(\delta^2).
$$
\fi
We choose $\delta > 0$ so that
$$
\frac{\sum_{(u,v)\in E_c} \mathrm{cap}_{uv}\cdot \|\tilde u - \tilde v\|^2}{\sum_{(u,v)\in E_d} \mathrm{dem}_{uv}\cdot \|\tilde u - \tilde v\|^2} < 1.
$$
We get that inequality~(\ref{ineq:vector}) holds. Finally, we verify that vectors $\set{\pm\tilde u}$ form an $\ell_2^2$ space.
We have,
\ifSODA
\begin{align*}
\|\tilde u - \tilde v\|^2 &{} + \|\tilde v - \tilde w\|^2 = 
\frac{\|\bar u - \bar v\|^2 + \|\bar v - \bar w\|^2 + 4\delta^2}{4R^2}\\
&\geq
\frac{\|\bar u - \bar w\|^2 + 4\delta^2}{4R^2} > \|\tilde u - \tilde w\|^2,\\
\|\tilde u - \tilde v\|^2 &{}+ \|\tilde v + \tilde w\|^2 =
\|\tilde u - \tilde v\|^2 +  (4 - \|\tilde v - \tilde w\|^2) \\
& > 
(\|\tilde u - \tilde w\|^2 + \|\tilde w - \tilde v\|^2) + (4 - \|\tilde v - \tilde w\|^2)\\
&= \|\tilde u + \tilde w\|^2,\\
\|\tilde u + \tilde v\|^2 &{}+ \|\tilde v + \tilde w\|^2 = 4 - \|\tilde u - \tilde v\|^2 + 4 - \|\tilde v - \tilde w\|^2\\
&\geq 8 - 2 > \|\tilde u - \tilde w\|^2.
\end{align*}
\else
\begin{align*}
\|\tilde u - \tilde v\|^2 + \|\tilde v - \tilde w\|^2 &= \frac{\|\bar u - \bar v\|^2 + \|\bar v - \bar w\|^2 + 4\delta^2}{4R^2}
\geq
\frac{\|\bar u - \bar w\|^2 + 4\delta^2}{4R^2} > \|\tilde u - \tilde w\|^2,\\
\|\tilde u - \tilde v\|^2 + \|\tilde v + \tilde w\|^2 &=
\|\tilde u - \tilde v\|^2 +  (4 - \|\tilde v - \tilde w\|^2) \\
&> 
(\|\tilde u - \tilde w\|^2 + \|\tilde w - \tilde v\|^2) + (4 - \|\tilde v - \tilde w\|^2)= \|\tilde u + \tilde w\|^2,\\
\|\tilde u + \tilde v\|^2 + \|\tilde v + \tilde w\|^2 &= 4 - \|\tilde u - \tilde v\|^2 + 4 - \|\tilde v - \tilde w\|^2
\geq 8 - 2 > \|\tilde u - \tilde w\|^2.
\end{align*}
\fi
In the the last line, we used that $\|\tilde u - \tilde v\|^2 = \left\|\frac{\bar u' - \bar c}{2R} - \frac{\bar v' - \bar c}{2R}\right\|^2 \leq 1$
for every $u$ and $v$.
\ifSODA
{\hfill\vbox{\hrule height.2pt\hbox{\vrule width.2pt height5pt \kern5pt \vrule width.2pt} \hrule height.2pt}}
\else
\end{proof}
\fi

\section{Remark on Correlation Clustering}\label{sec:corrclust}
%\ifSODA\scriptsize\fi
In this section, we briefly describe how our results extend to stable instances of the Correlation Clustering problem. The Correlation Clustering problem was introduced 
by Bansal, Blum, and Chawla~\cite{bbc} and later studied by Charikar and Wirth~\cite{CW}, and others. Our positive results for stable
and weakly stable instances of Max Cut also apply to stable and weakly stable instances of 2-Correlation Clustering.
Our negative result when Max $k$-Cut shows that there is no exact polynomial-time algorithm for $\infty$-stable instances 
of $k$-Correlation Clustering for $k \geq 3$ unless $RP=NP$.

\begin{Definition}
An instance of the $k$-Correlation Clustering problem is a weighted graph $G = (V,E,w)$ in which every edge 
is labeled with either ``$+$'' or ``$-$''. We denote the set of edges labeled with  ``$+$'' by $E^+$
and the set of edges labeled with ``$-$'' by $E^-$.
Consider a clustering $\cal C$ of $V$ into $k$ disjoint clusters. For every $u \in V$, let ${\cal C}(u)$ be the cluster that $u$ belongs to.
We define the total weight of agreements and disagreements as follows:
\begin{align*}
\mathsf{Agree}_G({\cal C}) &= \sum_{\substack{(u,v)\in E^+\\ {\cal C}(u) = {\cal C}(u)}} w_{(u,v)} + \sum_{\substack{(u,v)\in E^-\\ {\cal C}(u) \neq {\cal C}(u)}} w_{(u,v)}\\
\mathsf{DisAgree}_G({\cal C}) &= \sum_{\substack{(u,v)\in E^+\\ {\cal C}(u) \neq {\cal C}(u)}} w_{(u,v)} + \sum_{\substack{(u,v)\in E^-\\ {\cal C}(u) = {\cal C}(u)}} w_{(u,v)}.
\end{align*}

The Correlation Clustering problems asks to find a clustering of $V$ into $k$ clusters that maximizes the total weight of agreements, $\mathsf{Agree}_G({\cal C})$.
\end{Definition}

\begin{table*}
\centering
\begin{tabular}{|l|l|}
\hline \textbf{variant}     & \textbf{objective}\\ \hline
Maximum Agreement    & maximize  $\mathsf{Agree}_G({\cal C})$\\
Minimum Disagreement & minimize $\mathsf{DisAgree}_G({\cal C})$\\
Maximum Correlation  & maximize $\mathsf{Agree}_G({\cal C}) - \mathsf{DisAgree}_G({\cal C})$\\ \hline
\end{tabular}
\caption{The table shows different versions of the Correlation Clustering problem studied in the literature.}
\label{table:corrclust}
\end{table*}

We note that different variants of the problem have been studied in the literature (see Table~\ref{table:corrclust}). 
A good \textit{approximate} solution for one variant is not necessarily a good approximate solution for the other variants.
However, an \textit{optimal} solution for one variant is also an optimal solution for all other variants.
Thus an instance of Correlation Clustering is $\gamma$-stable w.r.t. one objective if and only if it is 
$\gamma$-stable w.r.t. each of them. Since in this paper we study exact algorithms for $\gamma$-stable instances, 
all three variants of the problem are equivalent for our purposes. We will assume that our objective is to maximize 
$\mathsf{Agree}_G({\cal C})$.

\paragraph{Positive Results.} Our positive results for stable and weakly stable instances of Max Cut also apply to stable and weakly stable instances of
$2$-Correlation Clustering. The proofs of Theorem~\ref{thm:sdp-integral}, Corollary~\ref{cor:algorithm}, and 
Theorem~\ref{thm:weakly-stable} can very easily be modified to deal with the $2$-Correlation Clustering problem.
We do not describe the necessary modifications in this paper. Instead, we point out that there is a simple reduction that maps
$\gamma$-stable instances of $2$ Correlation Clustering to $\gamma$-stable instances of Max Cut, and weakly stable instances to 
weakly stable instances. Therefore, every algorithm for solving $\gamma$-stable or $\gamma$-weakly stable instances of Max Cut
can be used to solve $\gamma$-stable or $\gamma$-weakly stable instances of 2-Correlation Clustering.
We now briefly describe the reduction.
Given a graph $G = (V,E^+ \cup E^-,w)$, the reduction constructs a graph $G'=(V', E', w')$ with
\ifSODA
\begin{itemize}
\item $V' = \set{u:u\in V} \cup \set{u':u\in V}$,
\item $E'  = \set{(u,v): (u,v) \in E^-} \cup \set{(u',u): u\in V} \cup \set{(u',v), (u,v'): (u,v) \in E^+}$,
\item $w'(u,v) = w(u,v)$, $w'(u,v') = w(u,v) /2$, $w(u,u') = W_{\infty}$,
\end{itemize}
\else
\begin{align*}
V' &= \set{u:u\in V} \cup \set{u':u\in V}, \\
E' & = \set{(u,v): (u,v) \in E^-} \cup \set{(u',u): u\in V} \cup \set{(u',v), (u,v'): (u,v) \in E^+},\\
w'(u,v) &= w(u,v), \quad w'(u,v') = w(u,v) /2, \quad w(u,u') = W_{\infty},
\end{align*}
\fi
where $W_{\infty}$ is large enough (e.g. $W_{\infty} = 2\gamma \sum_{e\in E} w_e$). Since the weight of edges $(u,u')$ is very large, every maximum cut in $G'$ cuts all edges $(u,u')$, even if we increase
some edge weights by a factor at most $\gamma$. For every 2-clustering $(S,\bar S)$ of $G$, consider the corresponding cut $(S', \bar S')$ in $G'$
\begin{align*}
S' &= \set{u: u\in S} \cup \set{u': u\in \bar S} \\
\bar S' &= \set{u: u\in \bar S} \cup \set{u': u\in S}.
\end{align*}
%Then $\mathsf{Agree}_G((S, \bar S))$ equals the weight of the cut $(S', \bar S')$ minus $nW_{\infty}$ (the latter term does not depend on $S$).
Then $\mathsf{Agree}_G((S, \bar S))= w'(E'(S', \bar S')) - nW_{\infty}$ (note that the term $nW_{\infty}$ does not depend on $S$).
We get that $(S,\bar S)$ is an optimal 2-clustering if and only if $(S', \bar S')$ is a maximum cut in $G'$.

\paragraph{Negative Results.} The Max $k$-Cut problem is a special case of the $k$-Correlation Clustering problem, in which all edges
are labeled with ``-''. Therefore, the result of Theorem~\ref{claim:neg-max-k-cut} applies to the $k$-Correlation Clustering problem when $k \geq 3$.

\section{Different Definitions of Weak Stability} \label{apx:equivalence}
In this section, we first prove that Definitions~\ref{def:weak-stabilityA} and~\ref{def:weak-stability} are equivalent,
and then that Definitions~\ref{def:weak-stability-multiwayA} and~\ref{def:weak-stability-multiway} are equivalent.
\begin{claim}Definitions~\ref{def:weak-stabilityA} and~\ref{def:weak-stability} are equivalent.
\end{claim}
\begin{proof}
Let $G=(V,E,w)$ be a $(\gamma,N)$-stable instance according to Definition~\ref{def:weak-stabilityA}. 
Let $(S,\bar S)$ be the maximum cut in $G$. Consider an arbitrary
cut $(T, \bar T)$ not in $N$. Define a $\gamma$-perturbation $G'$ of $G$ by
$w'(e) = \gamma w(e)$ if $e$ is cut by $(T,\bar T)$ and $w'(e) = w(e)$, otherwise.
Since $G$ is $\gamma$-stable, we have $w'(E(S,\bar S)) > w'(E(T,\bar T))$, and therefore,
$w'(E(S,\bar S)\setminus E(T,\bar T)) > w'(E(T,\bar T)\setminus E(S,\bar S))$. We conclude that
\ifSODA
\begin{multline*}
w(E(S,\bar S)\setminus E(T,\bar T)) = w'(E(S,\bar S)\setminus E(T,\bar T)) \\ > w'(E(T,\bar T)\setminus E(S,\bar S)) = \gamma w(E(T,\bar T)\setminus E(S,\bar S)).
\end{multline*}
\else
$$
w(E(S,\bar S)\setminus E(T,\bar T)) = w'(E(S,\bar S)\setminus E(T,\bar T)) > w'(E(T,\bar T)\setminus E(S,\bar S)) = \gamma w(E(T,\bar T)\setminus E(S,\bar S)).
$$
\fi
Therefore, $G$ is a $(\gamma, N)$-weakly stable instance according to Definition~\ref{def:weak-stability}.

Assume now that $G$ is a $(\gamma, N)$-weakly stable instance according to Definition~\ref{def:weak-stability}.
We need to show that for every $\gamma$-perturbation $G' = (V,E,w')$ and every cut $(T,\bar T)$, the following inequality holds:  $w'(E(S,\bar S)) > w'(E(T,\bar T))$, or, equivalently,
$w(E(S,\bar S)\setminus E(T,\bar T)) > w'(E(T,\bar T)\setminus E(S,\bar S))$.
We have,
\ifSODA
\begin{multline*}
w'(E(S,\bar S)\setminus E(T,\bar T)) \geq w(E(S,\bar S)\setminus E(T,\bar T)) \\ > \gamma w(E(T,\bar T)\setminus E(S,\bar S))
\geq w'(E(T,\bar T)\setminus E(S,\bar S)),
\end{multline*}
\else
$$
w'(E(S,\bar S)\setminus E(T,\bar T)) \geq w(E(S,\bar S)\setminus E(T,\bar T)) > \gamma w(E(T,\bar T)\setminus E(S,\bar S))
\geq w'(E(T,\bar T)\setminus E(S,\bar S)),
$$
\fi
as required.
\end{proof}

\ifSODA
\newpage
\fi

\begin{claim}Definitions~\ref{def:weak-stability-multiwayA} and~\ref{def:weak-stability-multiway} are equivalent.
\end{claim}
\begin{proof}
Let $G=(V,E,w)$ be a $(\gamma,N)$-stable instance according to Definition~\ref{def:weak-stabilityA}. 
Let ${\cal S}^*$ be the minimum multiway cut. Denote the set of edges cut by ${\cal S}^*$ by $E^*$.
Consider an arbitrary multiway cut ${\cal S}'$ not in $N$. Denote the set of edges cut by ${\cal S}'$ by $E'$.

Define a $\gamma$-perturbation $G'$ of $G$ by
$w'(e) = w(e)$ if $e\in E'$ and $w'(e) = \gamma w(e)$, otherwise.
Since $G$ is $\gamma$-stable, we have $w'(E^*) < w'(E')$, and therefore,
$w'(E^*\setminus E') < w'(E' \setminus E^*)$. We conclude that
$$\gamma w(E^*\setminus E') = w'(E^*\setminus E') < w'(E'\setminus E^*) = w(E'\setminus E^*).$$
Therefore, $G$ is a $(\gamma, N)$-weakly stable instance according to Definition~\ref{def:weak-stability-multiway}.

Assume now that $G$ is a $(\gamma, N)$-weakly stable instance according to Definition~\ref{def:weak-stability-multiway}.
We need to show that for every $\gamma$-perturbation $G' = (V,E,w')$ and every multicut ${\cal S}'$, the following inequality holds: $w'(E^*) < w'(E')$, or, equivalently,
$w(E^*\setminus E') < w'(E'\setminus E^*)$.
We have,
$$w'(E^*\setminus E') \leq \gamma w(E^*\setminus E') < w(E'\setminus E^*) \leq w'(E'\setminus E^*).$$
as required.
\end{proof}

\end{document}